\documentclass[%
reprint,
 amsmath,amssymb,
 aps,
pra,
]{revtex4-1}
\usepackage[english]{babel}
\usepackage{xcolor}

\usepackage{stmaryrd,acronym,amsthm,dsfont,graphicx,enumitem}
\setlist[itemize]{leftmargin=*}
\setlist[enumerate]{leftmargin=*}
\newcommand{\idle}{0}

\newcommand{\nidle}{\phi}
\newcommand{\nnu}{{m_x}}
\newcommand{\nnv}{{m_v}}

\newcommand{\puppm}[2]{\ket{\mathrm{PPM},#2}_{Q^{#1}}}

\newtheorem{theorem}{Theorem}
\newtheorem{lemma}{Lemma}
\newtheorem{definition}{Definition}
\newtheorem{proposition}{Proposition}
\newtheorem{remark}{Remark}

\DeclareMathAlphabet{\eurm}{U}{eur}{m}{n}
\DeclareMathAlphabet{\mathbsf}{OT1}{cmss}{bx}{n}
\DeclareMathAlphabet{\mathssf}{OT1}{cmss}{m}{sl}
\DeclareMathAlphabet{\mathcsf}{OT1}{cmss}{sbc}{n}

\DeclareSymbolFont{bsfletters}{OT1}{cmss}{bx}{n}  
\DeclareSymbolFont{ssfletters}{OT1}{cmss}{m}{n}
\DeclareMathSymbol{\bsfGamma}{0}{bsfletters}{'000}
\DeclareMathSymbol{\ssfGamma}{0}{ssfletters}{'000}
\DeclareMathSymbol{\bsfDelta}{0}{bsfletters}{'001}
\DeclareMathSymbol{\ssfDelta}{0}{ssfletters}{'001}
\DeclareMathSymbol{\bsfTheta}{0}{bsfletters}{'002}
\DeclareMathSymbol{\ssfTheta}{0}{ssfletters}{'002}
\DeclareMathSymbol{\bsfLambda}{0}{bsfletters}{'003}
\DeclareMathSymbol{\ssfLambda}{0}{ssfletters}{'003}
\DeclareMathSymbol{\bsfXi}{0}{bsfletters}{'004}
\DeclareMathSymbol{\ssfXi}{0}{ssfletters}{'004}
\DeclareMathSymbol{\bsfPi}{0}{bsfletters}{'005}
\DeclareMathSymbol{\ssfPi}{0}{ssfletters}{'005}
\DeclareMathSymbol{\bsfSigma}{0}{bsfletters}{'006}
\DeclareMathSymbol{\ssfSigma}{0}{ssfletters}{'006}
\DeclareMathSymbol{\bsfUpsilon}{0}{bsfletters}{'007}
\DeclareMathSymbol{\ssfUpsilon}{0}{ssfletters}{'007}
\DeclareMathSymbol{\bsfPhi}{0}{bsfletters}{'010}
\DeclareMathSymbol{\ssfPhi}{0}{ssfletters}{'010}
\DeclareMathSymbol{\bsfPsi}{0}{bsfletters}{'011}
\DeclareMathSymbol{\ssfPsi}{0}{ssfletters}{'011}
\DeclareMathSymbol{\bsfOmega}{0}{bsfletters}{'012}
\DeclareMathSymbol{\ssfOmega}{0}{ssfletters}{'012}

\newcommand{\calB}{{\mathcal{B}}}

\newcommand{\calD}{{\mathcal{D}}}
\newcommand{\calE}{{\mathcal{E}}}
\newcommand{\calF}{{\mathcal{F}}}

\newcommand{\calH}{{\mathcal{H}}}
\newcommand{\calI}{{\mathcal{I}}}

\newcommand{\calN}{{\mathcal{N}}}

\newcommand{\calP}{{\mathcal{P}}}

\newcommand{\calR}{{\mathcal{R}}}

\newcommand{\calU}{{\mathcal{U}}}
\newcommand{\calV}{{\mathcal{V}}}
\newcommand{\calX}{{\mathcal{X}}}
\newcommand{\calY}{{\mathcal{Y}}}
\newcommand{\calW}{{\mathcal{W}}}
\newcommand{\calZ}{{\mathcal{Z}}}

\renewcommand{\P}[2][]{{\mathbb{P}_{#1}}{\left(#2\right)}}
       
\newcommand{\D}[2]{{{\mathbb{D}}\!\left({#1\Vert#2}\right)}}

\newcommand{\avgI}[1]{{{\mathbb{I}}\!\left(#1\right)}}
\newcommand{\avgH}[1]{{\mathbb{H}}\!\left(#1\right)}

\newcommand{\Hb}[1]{{\mathbb{H}_b}\left(#1\right)}

\newcommand{\card}[1]{\ensuremath{\left|{#1}\right|}}           
\newcommand{\norm}[2][]{\ensuremath{{\left\Vert{#2}\right\Vert}_{#1}}}   
\newcommand{\eqdef}{\ensuremath{\triangleq}}                    
\newcommand{\intseq}[2]{\ensuremath{\llbracket{#1},{#2}\rrbracket}}  
\newcommand{\indic}[1]{\ensuremath{\mathds{1}\!\left\{#1\right\}}}

\renewcommand{\leq}{\leqslant}
\renewcommand{\geq}{\geqslant}

\newcommand{\tr}[1]{\ensuremath{\text{\textnormal{tr}}\left(#1\right)}}  

\newcommand{\proddist}{%
  \mathchoice{\raisebox{1pt}{$\displaystyle\otimes$}}
             {\raisebox{1pt}{$\otimes$}}
             {\raisebox{0.5pt}{\scalebox{0.7}{$\scriptstyle\otimes$}}}
             {\raisebox{0.4pt}{\scalebox{0.6}{$\scriptscriptstyle\otimes$}}}}
\newcommand{\pn}{{\proddist n}}

\acrodef{AEP}{Asymptotic Equipartition Property}
\acrodef{AoA}{Angle of Arrival}
\acrodef{AWGN}{Additive White Gaussian Noise}
\acrodef{AVC}[AVC]{Arbitrarily Varying Channel}
\acrodef{BER}{Bit-Error-Rate}
\acrodef{BEC}{Binary Erasure Channel}
\acrodef{BPSK}{Binary Phase-Shift Keying}
\acrodef{BSC}{Binary Symmetric Channel}
\acrodef{BICM}[BICM]{Bit-Interleaved Coded-Modulation}
\acrodef{CDF}[CDF]{Cumulative Distribution Function}
\acrodef{CGF}[CGF]{Cumulant Generating Function}
\acrodef{CLT}[CLT]{Central Limit Theorem}
\acrodef{CSI}[CSI]{Channel State Information}
\acrodef{DMC}[DMC]{Discrete Memoryless Channel}
\acrodef{DMS}[DMS]{Discrete Memoryless Source}
\acrodef{FER}[FER]{Frame Error Rate}
\acrodef{iid}[i.i.d.]{independent and identically distributed}
\acrodef{IoT}[IoT]{Internet of Things}
\acrodef{LPD}[LPD]{Low Probability of Detection}
\acrodef{LDPC}[LDPC]{Low-Density Parity-Check}
\acrodef{MAC}[MAC]{multiple-access channel}
\acrodef{MGF}[MGF]{Moment Generating Function}
\acrodef{MLC}[MLC]{Multi-Level Coding}
\acrodef{MIMO}[MIMO]{Multiple-Input Multiple-Output}
\acrodef{MISO}{Multiple-Input Single-Output}
\acrodef{MSD}[MSD]{Multi-Stage Decoding}
\acrodef{PDF}[PDF]{Probability Distribution Function}
\acrodef{PMF}[PMF]{Probability Mass Function}
\acrodef{PPM}[PPM]{Pulse Position Modulation}
\acrodef{PSD}{Power Spectral Density}
\acrodef{PSK}{Phase Shift Keying}
\acrodef{QKD}{Quantum Key Distribution}
\acrodef{CVQKD}{Continuous-Variable \ac{QKD}}
\acrodef{QPSK}{Quadrature Phase-Shift Keying}
\acrodef{SIMO}{Single-Input Multiple-Output}
\acrodef{SNR}{Signal-to-Noise Ratio}
\acrodef{TPCP}{Trace-Preserving Completely-Positive}
\acrodef{wrt}[w.r.t.]{with respect to}
\acrodef{WSS}{Wide Sense Stationary}

\newcommand{\pr}[1]{\left({#1}\right)}
\newcommand{\ket}[1]{|#1\rangle}
\newcommand{\bra}[1]{\langle #1 |}
\newcommand{\braket}[2]{\langle #1 | #2 \rangle}
\newcommand{\id}{\textnormal{id}}
\newcommand{\one}{\mathbf{1}}

\newcommand{\set}[1]{{\left\{#1\right\}}}
\newcommand{\ketbra}[2]{{\ket{#1}\bra{#2}}}
\newcommand{\kb}[1]{{\ket{#1}\bra{#1}}}
\newcommand{\ptr}[2]{\mathrm{tr}_{#1}\pr{#2}}
\newcommand{\mixed}[1]{\rho_{#1}^{\mathrm{unif}}}

\newcommand{\nrm}[1]{{\norm{#1}}}
\acrodef{POVM}{Positive Operator Valued Measurement}

\begin{document}

\title{Toward Undetectable Quantum Key Distribution over Bosonic Channels}
\author{Mehrdad Tahmasbi}
\thanks{Author to whom correspondence should be addressed. Email: mtahmasbi3@gatech.edu} 
\affiliation{Georgia Institute of Technology}

\author{Matthieu R. Bloch}
\thanks{Email: matthieu.bloch@ece.gatech.edu} 
\affiliation{Georgia Institute of Technology}


\begin{abstract}
  We show that covert secret key expansion is possible using a public  authenticated classical channel and a quantum channel largely under control of an adversary, which we precisely define. We also prove a converse result showing  that, under the golden standard of quantum key distribution by which the adversary completely controls the quantum channel, no covert key generation is possible.  We propose a protocol based on pulse-position modulation and multi-level coding that allows one to use traditional quantum key distribution (QKD) protocols while ensuring covertness, in the sense that no statistical test by the adversary can detect the presence of communication better than a random guess. When  run over a bosonic channel,  our protocol can leverage existing discrete-modulated continuous variable protocols. Since existing techniques to bound Eve's information do not directly apply, we develop a new bound that results in positive  throughput for a range of channel parameters.
\end{abstract}

\pacs{Valid PACS appear here}
\maketitle

\section{Introduction}
\label{sec:introduction}

The combination of quantum mechanics and information theory has led to several intriguing applications. In particular, there have been significant advances in {QKD}, which has now been successfully implemented and deployed in the field~\cite{Diamanti2016}. QKD finds its foundations in two pioneering papers~\cite{bennett1984quantum, Ekert1991}, which  discovered that non-classical signaling allows two parties (Alice and Bob) to exploit the laws of quantum mechanics and bound the information leaked to any adversary (Eve); when combined with classical information-theoretic tools, such as information reconciliation and privacy amplification, this observation can lead to protocols for the distillation of secure key bits. The security proofs of QKD have evolved from considering simple attacks, in which Eve could only perform a measurement on each transmitted  signal and send another state to Bob, to accounting for all attacks that could be described in the framework of quantum mechanics, known as \emph{coherent} attacks~\cite{renner2008security}; recent proofs even consider an adversary who tampers with the legitimate users' measurement devices~\cite{Acin2007}.

Although {QKD} ensures the \emph{confidentiality} of the generated keys in an extremely strong sense, Alice and Bob might desire other security features. One such feature that has recently attracted attention is  covertness~\cite{Bash2013,Wang2016b,Bloch2016a}, i.e., the ability to prevent an adversary from distinguishing whether a communication protocol is running or not from its observations.  For memory-less classical and classical-quantum (cq) channels, over which Alice aims at sending a message, a square root law has been established~\cite{Bash2013, Sheikholeslami2016} and states that the optimal number of bits that can be reliably and covertly transmitted scales as the square root of the number of channel uses. This contrasts with the limits of confidential communication, for which a linear scaling is feasible. The main intuition behind the square root law is that the central limit theorem ensures the presence of statistical uncertainty in Eve's observations, on the order of the square root of the number of channel uses, in which the transmitter can hide its signals.

The first attempts at covert QKD~\cite{Arrazola2016, Bash2015a} have ensured covertness with \emph{fully coordinated} protocols, in which information-bearing qubits are only transmitted over a secret random subset of channel uses upon which Alice and Bob secretly agree prior to communication; in the remaining channel uses Alice transmits an ``idle'' state corresponding to no communication. If $n$ denotes the total number of channel uses and $t$ denotes the number of channel uses over which transmission happens, fully coordinated protocols~\cite{Arrazola2016, Bash2015a} require $t=\Theta(\sqrt{n})$ to generate $\Omega(\sqrt{n})$ bits of secret key. Although the processing complexity is identical to that of standard QKD protocols, fully coordinated protocols require Alice and Bob to share $\log \binom{n}{t} = \Theta(\sqrt{n}\log n)$ secret bits prior to communication, so that the number of required key bit asymptotically dominates the number of generated bits. 

To circumvent the impossibility of key expansion with fully coordinated protocol, we have recently proposed~\cite{Tahmasbi2018b} to achieve covertness with an \emph{uncoordinated} protocol based on the use of ``sparse signaling'' for quantum state distribution. If $\alpha_n \eqdef O(n^{-\frac{1}{2}})$ and if $P_X$ denotes the Bernoulli($\alpha_n$) distribution, Alice  generates an i.i.d. sequence $X^n = (X_1, \cdots, X_n)$ according to $P_X^\pn$, which is then modulated by mapping zero to the idle state  and one to another state.
A technical subtlety, however, prevents Alice and Bob from performing classical information reconciliation and privacy amplification to obtain a secret key from their shared quantum states.   While the asymptotic key rate is $ O(n^{-\frac{1}{2}})$ by the square root law, the finite length penalty of privacy amplification is of the order of $\omega(n^{-\frac{1}{2}})$~\cite{Watanabe2013}, which  dominates the asymptotic rate. For a \emph{known adversary's attack}, our uncoordinated protocol circumvents this difficulty and ensures secret-key expansion using a likelihood encoder~\cite{Tahmasbi2018b} but the classical post-processing of the protocol \emph{is much more complex than for typical QKD protocols}.

To reap the benefits of both fully coordinated and uncoordinated protocols and achieve secret key expansion without increasing processing complexity, we develop here a \emph{partially coordinated} protocol inspired by our prior construction of low-complexity codes for covert communication over classical channels with {PPM} and {MLC}~\cite{Kadampot2018}. This approach is more aligned with traditional low-complexity information reconciliation and privacy amplification algorithms and we analyze the covertness and the security under an unknown attack by the adversary. We restrict, however, the adversary's  control of the channel by requiring that a portion of the channel be out of the adversary's control (e.g., the part of channel in Alice's laboratory). We prove that such a requirement is fundamentally necessary to establish any covertness result. We also point out that we were not able to use any standard technique to bound Eve's information. Accordingly, we present a new bound, by which we could in our security analysis achieve positive throughputs for some range of bosonic channel parameters. While our results are slightly disappointing in that this range of useful parameters is limited, they open the way to experimental demonstrations of covert QKD.

\section{Notation}
A system (e.g. $A$) is described by a finite-dimensional Hilbert space (e.g. $\calH_A$). Let $\one_A$ be the identity map on $\calH_A$ and $\mixed{A} \eqdef \frac{\one_A}{\dim \calH_A}$, where $\dim \calH_A$ is the dimension of $\calH_A$. $\calB(\calH_A)$ denotes the set of all bounded linear operators from $\calH_A$ to $\calH_A$, $\calP(\calH_A)$ denotes the set of all positive operators in $\calB(\calH_A)$, and $\calD(\calH_A)$ denotes the set of all density operators on $\calH_A$. For $X\in \calB(\calH_A)$, the trace norm of $X$ is ${\norm{X}}_1 \eqdef \tr{\smash{\sqrt{X^\dagger X}}}$, and $\nu(X)$ denotes the number of \emph{distinct} eigenvalues of $X$. We recall the definition of the von Neumann entropic quantities $H(\rho_A) \eqdef \avgH{A}_{\rho} \eqdef -\tr{\rho_A \log \rho_A}$, $\avgH{A|B}_{\rho} \eqdef \avgH{AB}_{\rho} - \avgH{B}_{\rho}$, and $\avgI{A;B}_{\rho} \eqdef \avgH{A}_{\rho} - \avgH{A|B}_{\rho}$. The fidelity between two density operators $\rho_A$ and $\sigma_A$ is defined as $F(\rho_A, \sigma_A) \eqdef {\norm{\smash{\sqrt{\rho_A} \sqrt{\sigma_A}}}}_1^2$. We further define $C(\rho_A, \sigma_A) \eqdef \sqrt{1- F(\rho_A, \sigma_A)}$, which satisfies the triangle inequality.  A quantum channel $\calN_{A\to B}$ is a linear trace-preserving completely positive map from $\calB(\calH_A)$ to $\calB(\calH_B)$. Let $\id_A$ be the identity channel on $\calB(\calH_A)$. For two states $\rho$ and $\sigma$, we define
\begin{align}
 \chi_2\pr{\rho\|\sigma} \eqdef \begin{cases} \tr{\rho^2\sigma^{-1}}-1 &\text{if }\, \text{supp} (\rho) \subset \text{supp} (\sigma),\\
 \infty&\text{otherwise}.\end{cases}
\end{align}
For a non-empty finite set $\calX$, let $\calH_X$ be a Hilbert space defined by an orthonormal basis $\set{\ket{x}: x\in \calX}$. For a function $f:\calX \to \calY$, we define the channel
\begin{align}
  \calE^f_{X\to Y}:\calB(\calH_X)&\to \calB(\calH_Y)\nonumber\\
  \rho_X &\mapsto \sum_{x\in \calX}  \ketbra{f(x)}{x}\rho_X \ketbra{x}{f(x)}.
\end{align}

%

\begin{figure}[b]
  \centering
  \includegraphics[width = 1 \linewidth]{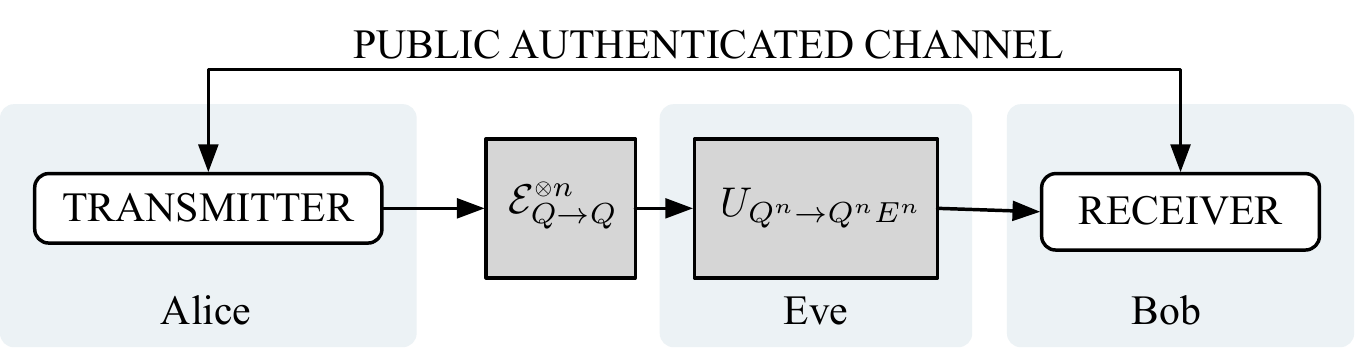}
  \caption{Covert quantum key expansion model}
  \label{fig:model}
\end{figure}

\section{Covert QKD Setup}
\label{sec:problem-formulation}

Alice and Bob aim  at \emph{covertly} expanding a \emph{secret} key in the following manner. Let $R_A$ and $R_B$ be Alice's and Bob's local randomness, respectively, and let $R$ be a secret common randomness. As depicted in Figure~\ref{fig:model}, Alice has a transmitter  in her lab to send quantum states to Bob. At any time instant, the state of the transmitter is described by a density operator on a Hilbert space $\calH_{Q}$. A pure state $\kb{0}$ identifies the ``idle'' state of the transmitter when there is no communication \footnote{One can associate a mixed state to no communication, but in bosonic systems, the natural choice for the idle state is a pure vacuum state.}. Alice prepares a quantum state $\widetilde{\sigma}_{AQ^n} = \ptr{RR_AR_B}{\widetilde{\sigma}_{RR_AR_BAQ^n}}$ and sends $\widetilde{\sigma}_{Q^n}$ to Bob by $n$ uses of her transmitter. The adversary Eve is assumed to receive the state through a known memoryless quantum  channel, which we call \emph{probe},  $\calE_{Q \to Q}$ that is \emph{outside its control}. Eve therefore obtains the output of $\calE_{Q \to Q}^\pn$ for the input $\widetilde{\sigma}_{Q^n}$, which then interacts with an ancilla $E^n$ in Eve's lab before being transmitted to Bob. The whole operation can be described by an isometry $U_{Q^n\to Q^n E^n}$, for which we denote the corresponding quantum channel by $\calU_{Q^n\to Q^n E^n}$. We call this phase  \emph{quantum state distribution}, which results in the joint quantum state
\begin{multline}
\sigma_{AQ^n E^n} \eqdef  \\(\id_{RR_AR_BA} \otimes \calU_{Q^n\to Q^n E^n }\circ \calE_{Q\to Q}^\pn )(\widetilde{\sigma}_{RR_AR_BA Q^n})
\end{multline}
between Alice, Bob, and Eve, respectively. After establishing a shared quantum state, Alice and Bob interactively communicate over an authenticated classical public channel and perform measurements on their available state to generate keys $S_A$ and $S_B$, respectively. We call this phase  \emph{quantum key distillation} and formally describe it by a quantum channel $\calD_{R_AR_BRAQ^n\to CS_AS_B}$, where $C$ denotes all public communication. The final state is then
\begin{multline}
\label{eq:final-state}
  \sigma_{CS_AS_BE^n} \eqdef \\
  (\id_{E^n} \otimes \calD_{R_AR_BRAQ^n\to CS_AS_B})(\sigma_{RR_AR_BAQ^n}).
\end{multline}
Furthermore, we assume that, in the absence of an adversary, Alice and Bob expect to be connected through the ``honest'' channel  $\calN_{Q\to Q}$ \emph{after the probe}. Alice and Bob can also abort the protocol at any time and do not generate secret keys. 
For a particular protocol inducing the final joint state $\sigma_{CS_AS_BE^n}$, we assess the  performance of the protocol with the following three quantities:
\begin{enumerate}
\item probability of error $\P{S_A \neq S_B|\mathrm{not~abort}}$;
\item  information leakage $\nrm{\sigma_{S_AE^nC}-\mixed{S_A}\otimes \sigma_{E^n C}}_1$;
\item covertness $\nrm{\sigma_{CE^n}-\mixed{C}\otimes\rho_{E^n}^0 }_1$, where $\rho_{E^n}^{0} \eqdef \calU_{Q^n \to E^n}(\calE_{Q\to Q}^\pn(\kb{0}^\pn))$; and
\item robustness $\P{\mathrm{abort}}$ in the presence of the honest channel $\calN_{Q\to Q}$.
\end{enumerate}

We highlight here three crucial distinctions between our model and traditional {QKD}.
\begin{enumerate}
\item  As covertness is of no concern in {QKD}, the  idle state of the transmitter is not specified in a {QKD} model. 
\item Unlike {QKD}, in which the quantum channel is in complete control of the adversary, we restrict Eve's observations to result from a known probe $\calE_{Q \to Q}$. We discuss this limit on our result in Section~\ref{sec:probe-role}.
\item Our covertness metric $\nrm{\sigma_{CE^n}-\mixed{C}\otimes\rho_{E^n}^0 }_1$ not only imposes a negligible dependence between public communication and $\sigma_{E^n}$ but also requires that public communication be distributed according to a pre-specified distribution, which we choose as the uniform distribution $\mixed{C}$ for simplicity. These two requirements are critical to ensure that public communication does not help Eve detect the communication.
\end{enumerate}

\section{Role of the Probe}
\label{sec:probe-role}
We establish here a no-go result in the absence of the warden's probe and for a relaxed  secrecy and covertness constraint.
\begin{theorem}
\label{th:prob-role}
Let $\calE_{Q\to Q} = \id_Q$ and define $K\eqdef \log \dim \calH_{S_A}$. Consider a protocol that operates as in Section~\ref{sec:problem-formulation} with $\P{S_A \neq S_B} \leq \epsilon$, $\nrm{\sigma_{S_AC} - \mixed{S_A}\otimes \sigma_C}_1\leq \delta$, and ${\norm{\widetilde{\sigma}_{Q^n} -\kb{0}^\pn}}_1 \leq \mu$. We then have
\begin{multline}
(1-5\sqrt{\mu} -\epsilon - 2\delta)K \leq 2\delta \log \dim \calH_C +  \Hb{\sqrt{\mu}} \\+ \Hb{\epsilon+\sqrt{\mu}} + 2\pr{1+\sqrt{\mu}}\Hb{\frac{\sqrt{\mu}}{1+\sqrt{\mu}}}.
\end{multline}
\end{theorem}
\begin{proof}
See Appendix~\ref{sec:no-go}.
\end{proof}
Consequently, if $\epsilon, \delta, \mu\to 0$, $K$ vanishes, as well.
Theorem~\ref{th:prob-role} therefore shows that giving the \emph{complete} control of the channel to the adversary is too stringent to establish covertness. A probe is therefore necessary and could by created with some part of the channel that is protected from the adversary, for example the portion of an optical fiber that lies inside Alice's lab.

\section{Protocol Description}
\label{sec:description-protocol}

We first provide a high level description of the role of PPM and MLC in our protocol. The principle of PPM is to split the whole transmission block into smaller sub-blocks and to transmit exactly one non-idle state in a position chosen uniformly at random in each sub-block. The number of sub-blocks and the size of each sub-block should both be $O(\sqrt{n})$ to achieve covertness~\cite{Bloch2017b}. 
The principle of MLC is to further split the randomness used to specify the position of the non-idle state into two parts: one part with a fixed size independent of $n$, generated locally by Alice and used for key generation, and another part of size growing with $n$, generated secretly and jointly by Alice and Bob and used for mimicking the uniform distribution via quantum channel resolvability \cite[Chapter 9.4]{hayashi2006quantum}. This splitting allows Alice and Bob to \emph{partially} coordinate  without paying the penalty incurred by full coordination. The use of {MLC} converts the problem of covert {QKD} into a traditional {QKD} problem over an effective block-length scaling as $O(\sqrt{n})$, for which low-complexity processing is possible.

\begin{figure}[t]
  \centering
  \includegraphics[width = 1 \linewidth]{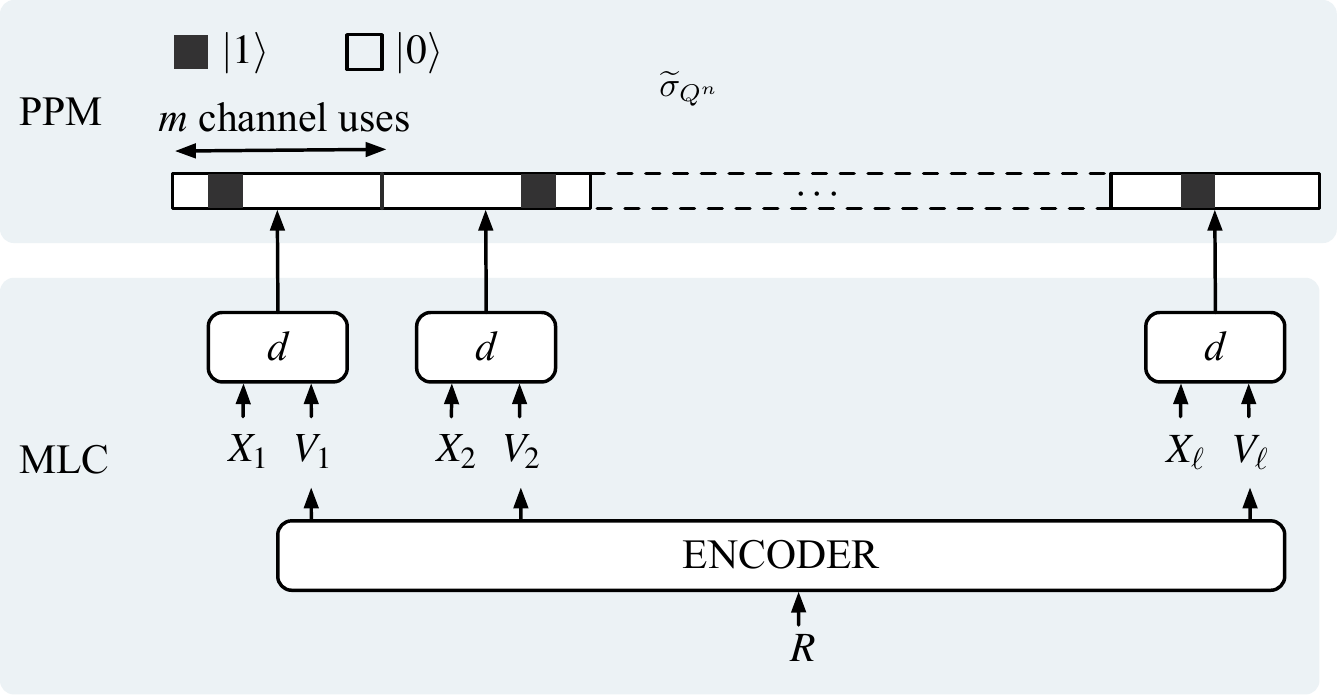}
  \caption{Covert quantum state distribution through {PPM} and {MLC}}
    \label{fig:scheme}

\end{figure} 

We now elaborate on the details of the partially coordinated protocol.  As depicted in Fig.~\ref{fig:scheme}, the $n$ channel uses are partitioned into $\ell$ consecutive sub-blocks of length $m$ so that $n \eqdef\ell m$.  Fix a non-idle state $\ket{\nidle}$ for the transmitter such that
\begin{align}
  \braket{\nidle}{0} &\neq 0\\
  \mathrm{supp}\pr{\calE_{Q\to Q}(\kb{1})} &\subseteq \mathrm{supp}\pr{\calE_{Q\to Q}(\kb{0})}
\end{align} We define the $z^{\mathrm{th}}$ PPM state of length $m$, $\puppm{m}{z}$ as 
\begin{align}
 {\ket{0}}^{\proddist z-1} \otimes \ket{\nidle} \otimes {\ket{0}}^{\proddist m-z},
\end{align} 
a product of $\ket{\idle}$ and $\ket{\nidle}$ with a single non-idle state in the $z^{\mathrm{th}}$ position. Writing $m\eqdef\nnu\nnv$, Alice generates  $\ell$  PPM states of length $m$ by choosing the position of the non-idle state in the $i^{\mathrm{th}}$ state as
\begin{align}
d(X_i, V_i) \eqdef (X_i-1)\nnv + V_i,
\end{align}
where $X^\ell  = (X_1, \cdots, X_\ell) \in \intseq{1}{\nnu}^\ell$ and $V^\ell = (V_1, \cdots, V_\ell) \in \intseq{1}{\nnv}^\ell$ are randomly generated sequences. Let $\rho_{Q^n}^{x^\ell, v^\ell}$ be the corresponding density operator when $X^\ell = x^\ell$ and $V^\ell = v^\ell$, i.e.,
\begin{align}
\rho_{Q^n}^{x^\ell, v^\ell} \eqdef \otimes_{i=1}^\ell  \kb{\mathrm{PPM}, d(x_i, v_i)}
\end{align}
The crux of the protocol is \emph{to generate the sequences $X^\ell$ and $V^\ell$ using different mechanisms:} $X^\ell$ is generated locally by Alice {i.i.d.} according to the uniform distribution over $\intseq{1}{m_x}$ while $V^\ell$ is generated jointly by Alice and Bob by sampling codewords uniformly at random from a codebook of size $h$ described as follows. Let $\calF$ be a regular two-universal family of hash functions from $\intseq{1}{\nnv}^\ell \to \calZ$ where $\calZ = \intseq{1}{\frac{\nnv^\ell}{h}}$. Bob samples  $f\in \calF$ and $z\in \calZ$ uniformly at random and transmits them over the public channel. The codebook consists of the codewords in $f^{-1}(z)$ and will be denoted through the function
\begin{align}
  g:\intseq{1}{h}\to \intseq{1}{\nnv}^\ell: R\mapsto V^\ell=g(R).
\end{align}
The choice of $X^\ell$ uniformly at random defines an effective cq channel from $v^\ell$ to the state at the output of the probe, formally described by 
\begin{align}
 v^\ell \mapsto \frac{1}{\nnu^\ell} \sum_{x^\ell} \calE_{Q\to Q}^{\pn}\pr{\rho^{x^\ell, v^\ell}_{Q^n}}.
\end{align} 
By sampling $R$ uniformly at random in $\intseq{1}{h}$ and using $g(R)$ at the input of the effective cq channel, Eve's received state is 
\begin{align}
  \sigma_{Q^n}\eqdef \frac{1}{\nnu^\ell}\frac{1}{h}\sum_{x^\ell, r}\calE_{Q\to Q}^{\pn}\pr{\rho_{Q^n}^{x^\ell, g(r)}}.\label{eq:state_protocol}
\end{align}
If Alice and Bob secretly share $R$ prior to the transmission, Bob can discard $m - \nnu$ of his sub-systems in each sub-block, for which he knows that the state $\ket{0}$ is sent. We shall later account for the partial coordination through $R$ by subtracting $\log h$ from the number of generated key bits. For each sub-block, Alice therefore obtains the classical state $X_i$ while Bob obtains $\nnu$ received states. We denote the whole state shared between Alice and Bob in $\ell$ sub-blocks by $\sigma_{X^\ell (Q^\nnu)^\ell}$, which is  $\tau_{XQ^{\nnu}}^{\proddist \ell}$ in the absence of the adversary for some $\tau_{XQ^{\nnu}}$ independent of $n$.

  The rest of the protocol is similar to a  traditional  QKD protocol applied to $\sigma_{X^\ell (Q^\nnu)^\ell}$ with the additional constraint $\nrm{\sigma_{CE^n}-\mixed{C}\otimes\rho_{E^n}^0 }_1\leq \delta$, which requires the public communication to  be uniformly distributed and  independent of Eve's observation during the quantum communication phase. The  three main steps of this phase are parameter estimation, information reconciliation, and privacy amplification. Let $\ell \eqdef \ell_1 + \ell_2$ and decompose $X^\ell(Q^\nnu)^\ell$ into two disjoint parts  $X^{\ell_1}(Q^\nnu)^{\ell_1}$ and $X^{\ell_2}(Q^\nnu)^{\ell_2}$, used for parameter estimation and secret key distillation, respectively. For simplicity, we do not detail the classical algorithm for information reconciliation and take for granted the existence of a protocol $\calI_{X^{\ell_2}(Q^\nnu)^{\ell_2} \to X^{\ell_2} \widehat{X}^{\ell_2} C_{\mathrm{IR}}}$ where $\widehat{X}^{\ell_2}$ denotes Bob's estimate of $X^{\ell_2}$ and $C_{\mathrm{IR}}$ is the public communication that takes place during the information reconciliation protocol. Let $\sigma_{X^{\ell_2}, \widehat{X}^{\ell_2} C_{IR}E^nC'	} \eqdef ( \calI \otimes \id_{E^nC'}) (\sigma_{X^{\ell_2} (Q^{\nnu})^{\ell_2} E^nC'})$ where $\sigma_{E^n}$ is the adversary's observation from the quantum communication and $C'$ is the public communication in the quantum state distribution phase. We \emph{assume} that
  \begin{align}
  \P{X^{\ell_2} \neq \widehat{X}^{\ell_2}} \leq \epsilon_{\mathrm{IR}},\\
    \P{\mathrm{abort} | \mathrm{~honest~channel}} \leq \epsilon_{\mathrm{IR}},\\
    \nrm{\sigma_{C_{\mathrm{IR}}E^nC'} - \mixed{C_{\mathrm{IR}} } \otimes \sigma_{E^nC'}}_1\leq \epsilon_{\mathrm{IR}}. \label{eq:add-inf-rec}
  \end{align}
  More justification of the existence of good reconciliation protocols can be found in \cite{Mateo2013} and references therein. Furthermore, using the ideas in \cite{Chou2016}, one can ensure the additional constraint in \eqref{eq:add-inf-rec}.
 The final step is to perform privacy application to establish a secure key. To this end, Alice and Bob require a bound on $\mathbb{H}_{\min}^{\delta_{\mathrm{PA}}}(X^{\ell_2}|E^n)$ for some information leakage threshold $\delta_{\mathrm{PA}}$, which we establish in Theorem~\ref{th:sec-ana}.

\section{Protocol Analysis}
\label{sec:protocol-analysis}
\subsection{Covertness}
\begin{theorem}
\label{th:cov-ana}
For any $\lambda_2 > 0$, with $$\log h =  \frac{\ell}{\nnu}\chi_2(\rho_E^1\|\rho_E^0)  + \sqrt{\ell} \pr{2\log \nnv +3} \sqrt{{\log \frac{4}{\lambda_2} + 1}},
$$ we have
\begin{align}
\nrm{\sigma_{E^n C} - \rho_{E^n}^0 \otimes \mixed{C}}_1 \leq \lambda_1 + \lambda_2 + \epsilon_{\mathrm{IR}} + \delta_{\mathrm{PA}},
\end{align}
where 
\begin{align}
\label{eq:lambda1}
\lambda_1 &= \sqrt{ \frac{\ell}{2m}\chi_2(\rho_E^1\| \rho_E^0)}.
 \end{align}
\end{theorem}
\begin{proof}
Let $C = (C', C'')$ where $C'$ denotes the public communication required to choose the codebook, and $C''$ denotes the remaining public communication. By the triangle inequality, we have
\begin{align}
&\nrm{\sigma_{E^n C} - \rho_{E^n}^0 \otimes \mixed{C}}_1\\
& \leq \nrm{\sigma_{E^n C} - \sigma_{E^nC'} \otimes \mixed{C''}}_1 \nonumber\\
&\phantom{=========}+ \nrm{\sigma_{E^nC'} \otimes \mixed{C''} - \rho_{E^n}^0 \otimes \mixed{C}}_1 \\
&=  \nrm{\sigma_{E^n C} - \sigma_{E^nC'} \otimes \mixed{C''}}_1 + \nrm{\sigma_{E^nC'}  - \rho_{E^n}^0 \otimes \mixed{C'}}_1
\end{align}
By our discussion at the end of Section~\ref{sec:description-protocol} and leftover hash lemma~\cite{renner2008security}, we have $ \nrm{\sigma_{E^n C} - \sigma_{E^nC'} \otimes \mixed{C''}}_1 \leq \epsilon_{\mathrm{IR}} + \delta_{\mathrm{PA}}$.
We now consider  the second term $\nrm{\sigma_{E^nC'}  - \rho_{E^n}^0\otimes \mixed{C'} }_1$. Note first that by the monotonicity of the trace norm,
\begin{align}
&\nrm{\sigma_{E^n}  - \rho_{E^n}^0 }_1\\
&= \nrm{\calU_{Q^n \to E^n}(\widetilde{\sigma}_{Q^nC'}) - \calU_{Q^n\to E^n}(\pr{\rho_Q^0}^\pn\otimes \mixed{C'})}_1\\
& \leq \nrm{\widetilde{\sigma}_{Q^nC'} - \pr{\rho_Q^0}^\pn\otimes \mixed{C'}}_1.
\end{align}
We upper-bound the above term in two steps. Introducing an intermediate state
\begin{align}
\label{eq:rho-ppm-def}
\rho_{Q^n}^{\mathrm{PPM}} \eqdef \frac{1}{\nnu^\ell\nnv^\ell}\sum_{x^\ell, v^\ell} \calE_{Q\to Q}^\pn(\rho_{Q^n}^{x^\ell, v^\ell}),
\end{align}
which is the average state at the output of the probe when $v^\ell$ is chosen uniformly at random from $\intseq{1}{\nnu}^\ell$, we have
\begin{align}
&\nrm{\rho_{Q^n}^{\mathrm{PPM}} \otimes \mixed{C'}  - \pr{\rho_Q^0}^\pn \otimes \mixed{C'}}\\
&\phantom{===========}= \nrm{\rho_{Q^n}^{\mathrm{PPM}}  - \pr{\rho_Q^0}^\pn}\\
&\phantom{===========}\stackrel{(a)}{\leq} \sqrt{\frac{1}{2}\D{\rho_{Q^n}^{\mathrm{PPM}}}{\pr{\rho_Q^0}^\pn}}\\
&\phantom{===========}\stackrel{(b)}{\leq} \sqrt{\frac{\ell}{2m} \chi_2(\rho_Q^1\|\rho_Q^0)}
\end{align}
where $(a)$ follows from Pinsker's inequality, and $(b)$ follows from \cite[Eq. (B144)]{tahmasbi2018framework}\footnote{In the classical setting, the authors of~\cite{Bloch2016b} showed the upper-bound with a factor of $1/2$ on the right hand side. While we conjecture that an extension of such upper-bound to the quantum setting is possible, we could only prove the upper-bound without the factor $1/2$}. Therefore, establishing covertness amounts to proving that the state $\sigma_{Q^nC'}$ generated by the protocol is nearly identical to $\rho_{Q^n}^{\mathrm{PPM}} \otimes \mixed{C'}$. 
This problem is known as quantum channel resolvability, and the minimum number of bits $\log h$ required is approximately equal to the Holevo information~\cite[Lemma 9.2]{hayashi2006quantum}. Recall that $\calF$ is a regular two-universal family of hash functions from $\intseq{1}{\nnv}^\ell$ to $\calZ \eqdef \intseq{1}{\nnv^\ell/h}$. 

Let us define
\begin{align}
\widetilde{\rho}^v_{Q^m} &\eqdef \frac{1}{\nnu} \sum_x \kb{\mathrm{PPM}, d(x, v)}\\
\widetilde{\rho}^{v^\ell}_{Q^n} &\eqdef \widetilde{\rho}^{v_1} \otimes \cdots \otimes \widetilde{\rho}^{v_\ell}\\
\widetilde{\rho}_{VQ^m} &\eqdef \frac{1}{\nnv} \sum_v \kb{v}_V \otimes \widetilde{\rho}^v_{Q^m}\\
\rho_{VQ^m} &\eqdef \calE_{Q\to Q}^{\proddist m}(\widetilde{\rho}_{VQ^m})
\end{align}

By Lemma~\ref{lm:res} in Appendix~\ref{sec:quantum-res},
\begin{align}
&\nrm{\sigma_{Q^nC'} - \rho_{Q^n}^{\mathrm{PPM}} \otimes \mixed{C'}}_1\\
&= \frac{1}{|\calF|} \frac{1}{\card{\calZ}} \sum_{f\in \calF, z\in \calZ}{\norm{\frac{1}{h} \sum_{v^\ell\in f^{-1}(z)} \calE_{Q\to Q}^\pn\pr{\rho_{Q^n}^{v^\ell}} - \rho_{Q^n}^{\mathrm{PPM}}}}_1\\
& \leq \lambda_2,
\end{align}
provided that 
\begin{multline}
\log h \geq \log \card{\calV^\ell} - \mathbb{H}^{\frac{\lambda}{4}}_{\min}(V^\ell|Q^n)_{\rho^{\proddist \ell}} + 2\log \frac{2}{\lambda_2},
 \end{multline}
and $\card{\calV}^\ell$ is divisible by $h$ \footnote{For any $h'$, we can choose $h$ such that $\card{\calV}h'\geq h\geq h'$ and $\card{\calV}^\ell$ is divisible by $h$; hence, this condition adds at most $\log \card{\calV} = O(\log n)$ of penalty on $\log h$.}. 
Applying~\cite[Corollary 3.3.7]{renner2008security}, we simplify the condition on $\log h$ by noting that
\begin{multline}
\log \card{\calV^\ell} - \mathbb{H}^{\frac{\lambda}{4}}_{\min}(V^\ell|Q^{n})_{\rho^{\proddist \ell}} \\
\begin{split}
&{\leq} \log  \card{\calV^\ell} - \ell\left(\avgH{V|Q}_{\rho} \right.\\\nonumber\\
&\left. \phantom{======}-\pr{2\mathbb{H}_{\max}(V)_{\rho} +3}\sqrt{\frac{\log \frac{4}{\lambda_2} + 1}{\ell}} \right)\\
&\stackrel{(a)}{=} \ell  \avgI{V; Q^m}_{\rho}  + \sqrt{\ell} \pr{2\log \nnv +3} \sqrt{{\log \frac{4}{\lambda_2} + 1}},
\end{split}
\end{multline}
where  $(a)$ follows since $\tau_V$ is the mixed state. We also further upper-bound $\avgI{V; Q^m}_{\rho}$ by
\begin{align}
\avgI{V; Q^m}_{\rho}  
&= \D{\rho_{VQ^m}}{\rho_V \otimes \rho_{Q^m}}\\
&\leq \D{\rho_{VQ^m}}{\rho_V \otimes \pr{\rho_Q^0}^{\proddist m}} \\
&=\frac{1}{\nnv}  \sum_{v\in \calV}\D{\rho_{Q^m}^v}{\pr{\rho_Q^0}^{\proddist m}}\\
&\stackrel{(a)}{=}  \D{\rho_{Q^m}^1}{\pr{\rho_Q^0}^{\proddist m}}\\
&= \D{\calE_{Q\to Q}^{\proddist m_x}\pr{\rho^{\mathrm{PPM}}_{Q^{\nnu}}}}{\pr{\rho_Q^0}^{\proddist \nnu}}\\
&\leq \frac{1}{\nnu} \chi_2(\rho_Q^1\|\rho_Q^0),
\end{align}
where $(a)$ follows from the symmetry in the definition of $\rho^v_{Q^m}$.
This concludes the proof.

\end{proof}

\subsection{Security}
The objective of this section is to lower bound the smooth min-entropy of Alice's data $X^\ell$ given Eve's observations. We first remind that, by our discussion in Section~\ref{sec:description-protocol}, we can assume that Alice prepares $\widetilde{\sigma}_{XQ^\nnu}^{\proddist \ell}$ where 
\begin{align}
\widetilde{\sigma}_{XQ^\nnu} &\eqdef \frac{1}{\nnu} \sum_{x=1}^{\nnu} \kb{x}_X \otimes \widetilde{\sigma}_{Q^\nnu}^x,\\
\widetilde{\sigma}_{Q^\nnu}^x &\eqdef {\kb{\idle}}^{\proddist x - 1} \otimes \kb{\nidle} \otimes {\kb{\idle}}^{\proddist \nnu - x},
\end{align}
and sends $\widetilde{\sigma}_{Q^\nnu}^{\proddist \ell}$ over the quantum channel to Bob. We assume that Eve applies the same unitary $U_{Q^{\nnu} \to Q^{\nnu}  E^m }$ on each PPM symbol. Generalizing the security proof to a general attack could follow from the same techniques as in \cite{Franz2012, Leverrier2013, Renner2009}. We now introduce some notation, which is summarized in Fig.~\ref{fig:states}.
\begin{figure}[t]
  \centering
  \includegraphics[width = \linewidth]{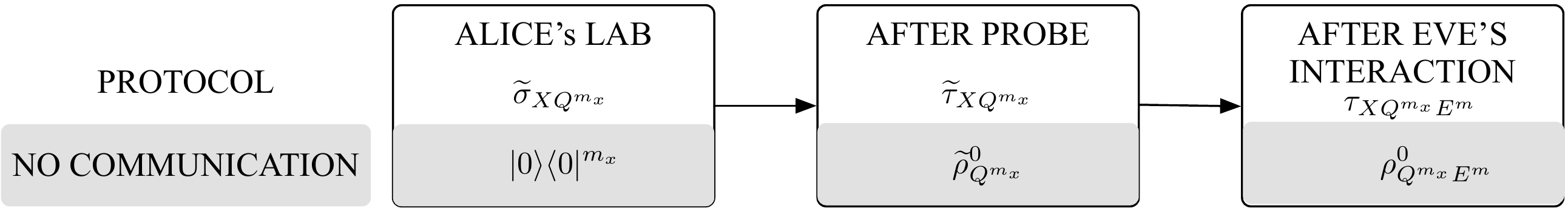}
  \caption{Notations for secrecy analysis}
    \label{fig:states}

\end{figure} 
 Let us define
 \begin{align}
 \widetilde{\tau}_{Q^\nnu}^x &\eqdef \calE_{Q\to Q}^\nnu\pr{\widetilde{\sigma}_{Q^\nnu}^x},\\
 \tau_{Q^\nnu E^m}^x &\eqdef U_{Q^\nnu \to Q^\nnu\to E^m} \widetilde{\tau}_{Q^\nnu}^xU_{Q^\nnu \to Q^\nnu\to E^m}^\dagger,\\
\tau_{X Q^\nnu E^m} &\eqdef \frac{1}{\nnu} \sum_{x=1}^\nnu \kb{x}_X \otimes  \tau_{Q^\nnu E^m}^x,\\
\sigma_{X^\ell Q^{\nnu \ell} E^n} &\eqdef \tau_{X{B}E^m}^{\proddist \ell}.
\end{align}
We also define
\begin{align}
\widetilde{\rho}_{Q^\nnu}^0 &\eqdef  \pr{ \calE_{Q\to Q}\pr{\kb{0}}}^{\proddist \nnu},\\
{\rho}_{Q^\nnu E ^m}^0 &\eqdef U_{Q^\nnu \to Q^\nnu\to E^m} \widetilde{\rho}_{Q^\nnu}^0U_{Q^\nnu \to Q^\nnu\to E^m}^\dagger.
\end{align}

\begin{theorem}
\label{th:sec-ana}
We have
\begin{multline}
\label{eq:bound-eve-inf}
\mathbb{H}_{\min}^\delta(X^\ell|E^n)_{\sigma} \geq \log \nnu -  \D{\rho_Q^1}{\rho_Q^0} \\+ \frac{1}{\nnu}\sum_{x} \log\pr{1 - \eta^x}- \pr{2\log \nnu + 3}\sqrt{\frac{\log \frac{1}{\delta}+ 1}{\ell}},
\end{multline}
for all unitaries $U$ and for all $\eta^x$ such that
\begin{multline}
\label{eq:cond-th3}
F(\tau_{Q^\nnu}^x, \rho_{Q^\nnu}^0) \leq \aleph(\lambda^x, F(\kb{\idle}, \kb{\nidle}))   \\- 2\sqrt{1-F(\kb{\idle}, \kb{\nidle})}\delta -(\delta)^2
\end{multline}
where
\begin{align}
\aleph(x, y) &\eqdef 1 - \frac{2\sqrt{1-y}x + x^2}{y}- 2\sqrt{\frac{2\sqrt{1-y}x + x^2}{y}}x-x^2,\\
\lambda^x &\eqdef 2\delta_0 + \sqrt{\eta^x + 4\sqrt{\eta^x}\delta_1 + 4(\delta_1)^2},\\
\delta_1 &\eqdef C(\kb{\nidle}, \rho_Q^1),\\
\delta_0 &\eqdef C(\kb{0}, \rho_Q^0),\\
\delta &\eqdef \delta_0 + \delta_1.
\end{align}
\end{theorem}
\begin{remark}
The right hand side of \eqref{eq:bound-eve-inf} only depends on quantities that are either specified by the protocol and the probe, or could be calculated from Alice's and Bob's observations.
\end{remark}
\begin{remark}
We explain here the difficulty in obtaining such bounds.   Note first that, as detailed in \cite{Tahmasbi2018b}, reverse reconciliation does not lead to a positive covert throughput unless Eve's and Bob's observations are independent when $\ket{0}$ is sent. This is unfortunately not  the case when the channel is a beam-splitter.  To the best of our knowledge, there exist two standard methods to bound Eve's information for continuous variable QKD protocols. The first method leverages the optimality of Gaussian attack, which results in a sub-optimal bound on Eve's information for discrete-variable protocols. Since Alice's measurement is not Gaussian (in the entanglement-based version), it is not straightforward to calculate the bound for  forward reconciliation protocols. The second method exploits entropic uncertainty relations, which requires finding an entanglement-based version with two different measurement at Alice. We could not find such version of our specific quantum state distribution.
\end{remark}
\begin{remark}
Note that in the absence of the adversary $\avgI{X; Q^{\nnu}}_{\sigma} = \D{\calN(\rho_Q^1)}{\calN(\rho_Q^0)} + O(1/\nnu)$ \cite{Bloch2016b}. Excluding finite-length effects, we achieve positive covert throughput when,
\begin{align}
\D{\rho_Q^1}{\rho_Q^0} -  \D{\calN(\rho_Q^1)}{\calN(\rho_Q^0)} \leq \frac{1}{\nnu} \sum_x \log (1 - \eta_x).
\end{align}
This inequality holds when $\eta_x> 0$ and $\calN$ is close to the noiseless channel.
\end{remark}
We now state a general upper bound for the relative entropy between the output of the complementary channel for two fixed states.
\begin{theorem}
\label{th:data-proc-ref}
Let $A$ and $B$ be two possibly infinite dimensional  quantum systems such that system $A$ is a composition of two sub-systems $A'$ and $A''$. Let $\rho_A^0$ and $\rho_A^1$ be in $\calD(\calH_A)$ such that for two pure states $\ket{\phi^0}_{A'}$ and $\ket{\phi^1}_{A'}$ in $\calH_{A'}$ and a mixed state $\nu_{A''}$ in $\calD(\calH_{A''})$, we have $C(\phi^x_{A'} \otimes \nu_{A''}, \rho^x_A)\leq \delta_x$.   Let $\calN:\calD(\calH_A) \to \calD(\calH_B)$ be a quantum channel with a complementary channel $\calE:\calD(\calH_A) \to \calD(\calH_E)$. Suppose that $\eta >0$ satisfies
\begin{multline}
\label{eq:eta-assump}
F(\calN(\rho^1_A), \calN(\rho^0_A)) \leq \aleph(\lambda, F(\phi^1_{A'}, \phi^0_{A'}))  \\ - 2\sqrt{1-F(\phi^1_{A'}, \phi^0_{A'})}\delta -\delta^2
\end{multline}
where $\lambda \eqdef 2\delta_0 + \sqrt{\eta + 4\sqrt{\eta}\delta_1 + 4\delta_1^2}$, $\delta \eqdef \delta_0 + \delta_1$.

 We then have
\begin{align}
\D{\calE(\rho_A^1) }{\calE(\rho_A^0)} \leq \D{\rho_A^1}{\rho_A^0} + \log\pr{1-\eta}.
\end{align}
\end{theorem}

\begin{proof}
See Appendix~\ref{sec:data-proc-ref}.
\end{proof}

\begin{proof}[Proof of Theorem~\ref{th:sec-ana}]

 By \cite[Corollary 3.3.7]{renner2008security}, we have
\begin{align}
\mathbb{H}_{\min}^\epsilon\pr{X^\ell | E^n}_{\sigma} 
&\geq \avgH{X|E}_{\tau} - \pr{2\mathbb{H}_{\max}\pr{X}_\tau + 3}\sqrt{\frac{\log \frac{1}{\epsilon}+ 1}{\ell}}\\
&=  \avgH{X|E^m}_{\tau} - \pr{2\log \nnu + 3}\sqrt{\frac{\log \frac{1}{\epsilon}+ 1}{\ell}}.
\end{align}
Furthermore,
\begin{align}
\avgH{X|E^m}_{\tau} = \avgH{X}_\tau  - \avgI{X;E^m}_\tau = \log \nnu -  \avgI{X;E^m}_\tau.
\end{align}
Note now that
\begin{align}
 \avgI{X;E^m}_\tau 
 &= \D{\tau_{XE^m}}{\tau_X \otimes \tau_{E^m}}\\
 &\leq \D{\tau_{XE^m}}{\tau_X \otimes \rho_{E^m}^0}\\
 &= \frac{1}{\nnu}\sum_{x=1}^{\nnu} \D{\tau_{E^m}^x}{\rho_{E^m}^0}.
\end{align}
Since $\eta_x$ satisfies the condition in \eqref{eq:eta-assump} by \eqref{eq:cond-th3}, We can apply Theorem~\ref{th:data-proc-ref} to obtain
\begin{align}
\D{\tau_{E^m}^x}{\rho_{E^m}^0} \leq \D{\widetilde{\tau}^x_{Q^\nnu}}{\widetilde{\rho}^0_{Q^\nnu}} + \log\pr{1 - \eta_x}.
\end{align}
Combining the above inequalities, we obtain the result.
\end{proof}

\subsection{Example}
We present here an experimental setup over which our proposed scheme could be executed. As illustrated in Fig.~\ref{fig:exp}, Alice's transmitter is a laser  whose output is a single-mode bosonic system. The idle state is $\ket{0}$ and we choose a \emph{coherent} state $\ket{\alpha}$ as the non-idle state. The probe and the honest channel are both  beam-splitters with transmissivity $\tau_E$ and $\tau_N$, respectively, and excess noise $\overline{n}_E$ and $\overline{n}_N$, respectively. In Fig.~\ref{fig:example}, we plot the number of bits per PPM symbol versus $\tau_N$ for $\tau_E = 0.9994$, $\alpha = 0.6$, $\overline{n}_E = 11$, and $\overline{n}_N = 0.01$. For these parameters, we also have $\chi_2(\calE(\kb{0})\| \calE(\kb{\alpha})) = 59881934$, which controls the covertness through Eq.~\eqref{eq:lambda1}.

Although the range of channel parameters highlighted is narrow and the efficiency is very low, this example shows the possibility of covert QKD in settings not envisioned earlier . One can certainly improve the performance of the protocol by developing tighter bound for Eve's information, which we leave out for future investigations.
\begin{figure}[t]
  \centering
  \includegraphics[width = \linewidth]{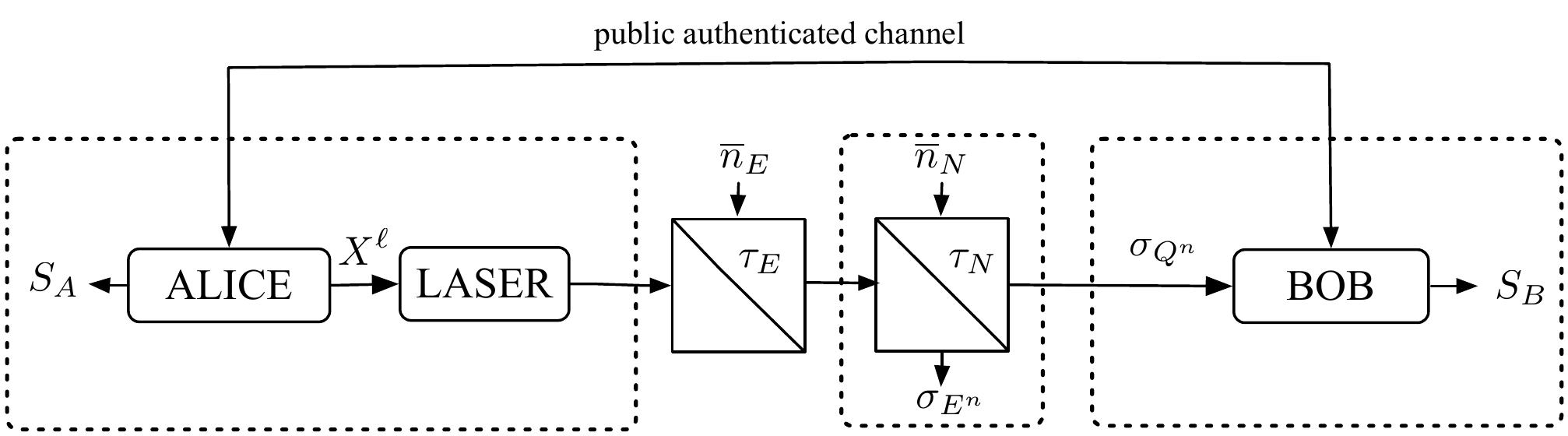}
  \caption{Experimental setup for our protocol.}
    \label{fig:exp}

\end{figure} 
\begin{figure}[t]
  \centering
  \includegraphics[width = \linewidth]{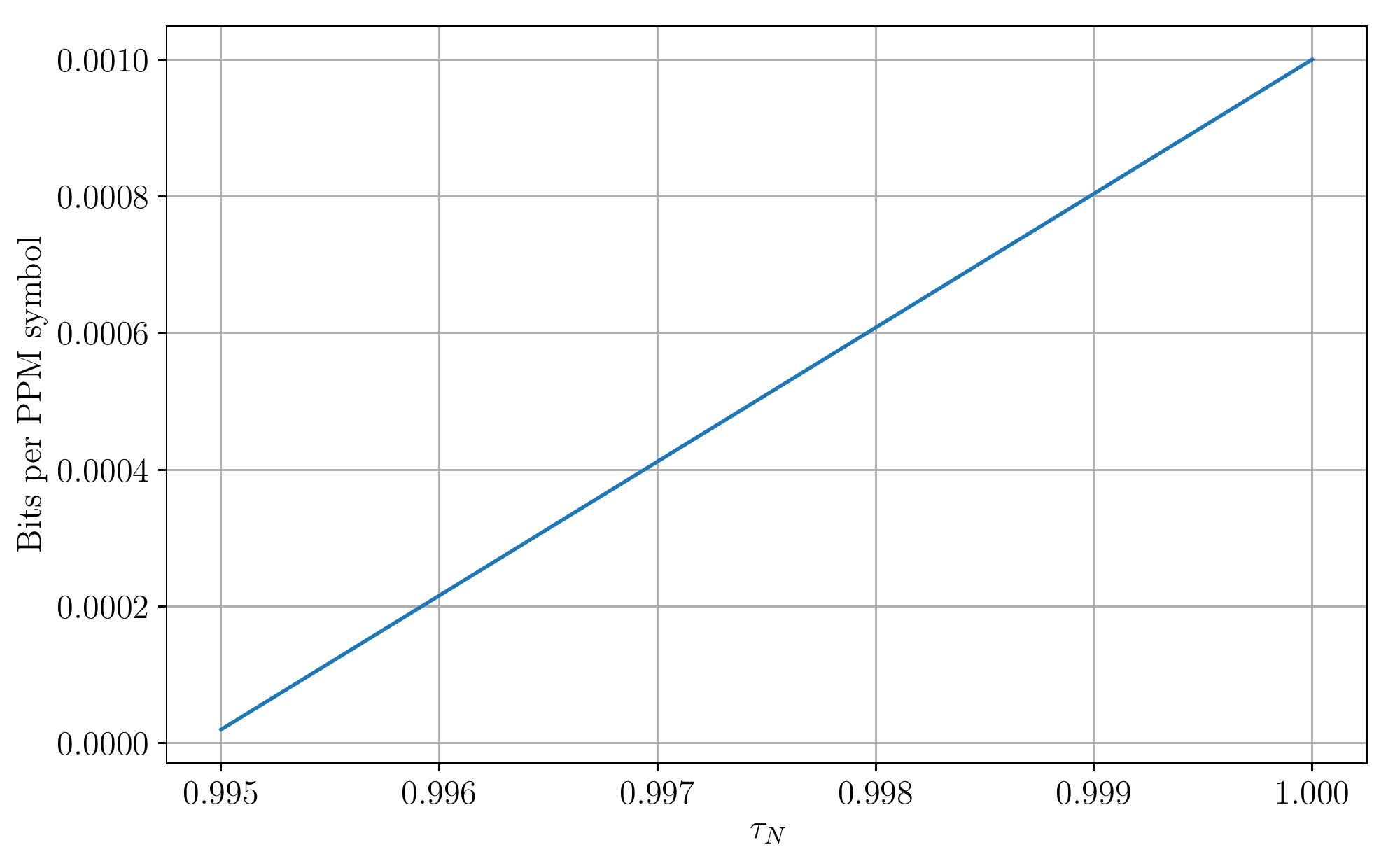}
  \caption{Achievable number of key bits per PPM symbol.}
    \label{fig:example}

\end{figure} 
%


 \section*{acknowledgement}
This work is supported by NSF under Award 1910859.
\appendix
\section{Proof of Theorem~\ref{th:prob-role}}
\label{sec:no-go}
 We first prove a quantum counterpart of \cite[Lemma 2.2]{Ahlswede1993}.
\begin{lemma}
\label{lm:pub-com}
Let $\rho_{AB}$ be a bipartite state and $\calE_{A\to A F}$ be a quantum channel. We then have
\begin{align}
\avgI{A;B}_\rho \geq \avgI{A;B|F}_{\rho'},
\end{align}
where $\rho'_{ABF} \eqdef (\calE_{A\to AF}\otimes \id_B)(\rho_{AB})$.
\end{lemma}
\begin{proof}

We have
\begin{align}
\avgI{A;B|F}_{\rho'} 
&= \avgH{B|F}_{\rho'} - \avgH{B|FA}_{\rho'}\\
&\stackrel{(a)}{\leq} \avgH{B}_{\rho'} - \avgH{B|FA}_{\rho'}\\
&= \avgI{AF;B}_{\rho'}\\
&= \D{\rho'_{ABF}}{ \rho'_{AF} \otimes \rho'}\\
&\stackrel{(b)}{\leq} \D{\rho_{AB}}{\rho_{A} \otimes \rho_B} = \avgI{A;B}_{\rho},
\end{align}
where $(a)$ follows from the sub-additivity of the von Neumann entropy, and $(b)$ follows from the data processing inequality.
\end{proof}
\begin{proof}[Proof of Theorem~\ref{th:prob-role}]
Let $\widetilde{\sigma}_{AQ^n}$ be the state initially prepared by Alice such that ${\norm{\widetilde{\sigma}_{Q^n} - \kb{0}^\pn}}_1\leq \mu$. We then have
\begin{align}
F(\widetilde{\sigma}_{Q^n}, \kb{0}^\pn) \geq 1- \mu.
\end{align} 
Let $\widetilde{\sigma}_{RAQ^n}$ be a purification of $\widetilde{\sigma}_{AQ^n}$. By Uhlmann's theorem, there exists a unit vector $\ket{\phi}_{RA}$ such that
\begin{align}
F(\widetilde{\sigma}_{RAQ^n}, \phi_{RA}\otimes \kb{0}^\pn) \geq 1-\mu.
\end{align}
Let $\widetilde{\tau}_{AQ^n} \eqdef \ptr{R}{\kb{\phi}_{RA}}\otimes \kb{0}^\pn$ and $\tau_{CS_AS_BE^n}$ be the output of the same protocol if Alice initially prepares $\widetilde{\tau}_{AQ^n} $ instead of $\widetilde{\sigma}_{AQ^n} $. By monotonicity of the fidelity, we have $F(\widetilde{\tau}_{AQ^n}, \widehat{\sigma}_{AQ^n})\geq 1-\mu$, and therefore, ${\norm{\widetilde{\tau}_{AQ^n}  - \widetilde{\sigma}_{AQ^n} }}_1\leq \sqrt{\mu}$. By the data processing inequality, we also have ${\norm{\tau_{CS_AS_BE^n} - \sigma_{CS_AS_BE^n}}}_1\leq \sqrt{\mu}$. This implies that $\P{S_A \neq S_B}_\tau \leq \epsilon + \sqrt{\mu}$. By \cite[Exercise 11.10.2]{wilde2013quantum}, we also have
\begin{align}
&\avgI{S_A;C}_\tau \\
&\leq \avgI{S_A;C}_\sigma +  3\sqrt{\mu} K + 2(1+\sqrt{\mu}) \Hb{\frac{\sqrt{\mu}}{1+\sqrt{\mu}}}\\
&\stackrel{(a)}{\leq} \delta\pr{K + \log \dim \calH_C} +  3\sqrt{\mu} K\nonumber\\
&\phantom{===========} + 2(1+\sqrt{\mu}) \Hb{\frac{\sqrt{\mu}}{1+\sqrt{\mu}}},
\end{align}
where $(a)$ follows from
\begin{multline}
\D{\rho_{S_AC}}{\mixed{S_A}\otimes \rho_C}\\ \leq \nrm{\rho_{S_AC}- \mixed{S_A}\otimes \rho_C}_1 \pr{K + \log \dim \calH_C}.
\end{multline}
By  Fannnes's inequality,
\begin{align}
\avgH{S_A}_\sigma \leq \avgH{S_A}_\tau + \sqrt{\mu}K + \Hb{\sqrt{\mu}}.
\end{align}
Note that
\begin{align}
K
&= \avgH{S_A}_\sigma + \D{\sigma_{S_A}}{\mixed{S_A}} \\
&\leq  \avgH{S_A}_\sigma  + \delta\pr{K + \log \dim \calH_C}\\
&\leq \avgH{S_A}_\tau +  \sqrt{\mu}K + \Hb{\sqrt{\mu}} + \delta\pr{K + \log \dim \calH_C}. \label{eq:converse1}
\end{align}
We furthermore have
\begin{align}
\avgH{S_A}_\tau &= \avgH{S_A|C}_\tau + \avgI{S_A;C}_\tau\\
&\leq  \avgH{S_A|C}_\tau  + \delta\pr{K + \log \dim \calH_C}\nonumber\\
&\phantom{=} +  3\sqrt{\mu}K + 2(1+\sqrt{\mu}) \Hb{\frac{\sqrt{\mu}}{1+\sqrt{\mu}}}. \label{eq:converse2}
\end{align}
Using  Fano's inequality, we obtain
\begin{align}
&\avgH{S_A|C}_\tau \\
&\leq  \avgI{S_A;S_B|C}_\tau +  \Hb{\epsilon + \sqrt{\mu}} + \pr{\epsilon + \sqrt{\mu}}K\\
&\stackrel{(a)}{\leq}  \avgI{A;Q^n}_\tau +  \Hb{\epsilon + \sqrt{\mu}} + \pr{\epsilon + \sqrt{\mu}}K\\
&\stackrel{(b)}{\leq} \avgI{A;Q^n}_{\widetilde{\tau}} +   \Hb{\epsilon + \sqrt{\mu}} + \pr{\epsilon + \sqrt{\mu}}K\\
&\stackrel{(c)}{=} \Hb{\epsilon + \sqrt{\mu}} + \pr{\epsilon + \sqrt{\mu}}K, \label{eq:converse3}
\end{align}
where $(a)$ follows from using Lemma~\ref{lm:pub-com} for each use of the public channel, $(b)$ follows from data processing inequality, and $(c)$ follows since $\widetilde{\tau}_{AQ^n} = \widetilde{\tau}_A \otimes \widetilde{\tau}_{Q^n}$. Combining \eqref{eq:converse1}, \eqref{eq:converse2}, and \eqref{eq:converse3}, we obtain the desired bound. 
\end{proof}

\section{A Quantum Resolvability Result}
\label{sec:quantum-res}
We prove a quantum channel resolvability result  based on the privacy amplification result of~\cite{renner2008security}. Note that we cannot use the standard quantum resolvability result of~\cite{hayashi2006quantum} since it depends on the dimension of the output space, which itself grows exponentially for $v^\ell \mapsto \rho_{Q^n}^{v^\ell}$. We first recall the definition of two-universal family of hash functions.
\begin{definition}
Let $\calX$ and $\calZ$ be two finite non-empty sets. A non-empty family of functions $\calF$ from $\calX$ to $\calZ$ is called \emph{two-universal} if for all distinct $x, x'\in \calX$, we have
\begin{align}
\frac{1}{\card{\calF}}\sum_{f\in \calF} \indic{f(x) = f(x')} \leq \frac{1}{\card{\calZ}}.
\end{align}
Moreover, $\calF$ is called \emph{regular} if for all $f\in \calF$ and all $z\in \calZ$, we have $\card{f^{-1}(z)} = \frac{\card{\calX}}{\card{\calZ}}$, where $f^{-1}(z) \eqdef \set{x\in \calX: f(x) = z}$.
\end{definition}
The next two results are well-known properties of two-universal hash functions.
\begin{proposition}
\label{lm:hash-exist}
Let $\calX$ and $\calZ$ be two non-empty finite sets such that $\card{\calX}$ is divisible by $\card{\calZ}$. There exists a two-universal regular family of functions from $\calX$ to $\calZ$.
\end{proposition}
\begin{proof}
All functions $f$ with $f^{-1}(z) = \frac{\card{\calZ}}{\card{\calX}}$ form a two-universal family of hash functions.
\end{proof}
\begin{proposition}[~\cite{renner2008security}]
\label{lm:priv-amp}
 Let $\rho_{XA}$ be a cq state on $\calH_X \otimes \calH_A$ with respect to an orthonormal basis $\set{\ket{x}: x\in \calX}$ for $\calH_X$, and $\calF$ be a two-universal family of functions from $\calX$ to $\calZ$. We then have
\begin{multline}
\frac{1}{\card{\calF}}\sum_{f\in \calF}{\norm{(\calE^f_{X\to Z}\otimes\id_A) (\rho_{XA})-\mixed{Z}\otimes \rho_A  }}_1 \\
\leq \inf_{\epsilon\geq 0}\left[ 2\epsilon + 2^{-\frac{1}{2} \pr{\mathbb{H}^{\epsilon}_{\min}(X|A)_{\rho} - \log \card{\calZ} } }\right]
\end{multline}
\end{proposition}

We are now ready to establish the main result of this section, which shows the existence of a resolvability code. The classical counter-part of this result was proved in~\cite{Kadampot2018}.
\begin{lemma}
\label{lm:res}
Let $\rho_{XA} = \sum_{x\in \calX} \frac{1}{|\calX|} \ket{x}\bra{x} \otimes \rho_A^x$ be a cq state on $\calH_X\otimes \calH_A$. Let $\delta>0$, $h$ be a positive integer such that $\card{\calX}$ is divisible by $h$ and
\begin{align}
\label{eq:logk-bound}
\log h \geq \log \card{\calX} - \mathbb{H}_{\min}^{\frac{\delta}{4}}(X|A)_{\rho} +2 \log \frac{2}{\delta}.
\end{align}
For a regular two-universal family of hash functions $\calF$ from $\calX$ to $\calZ$, we have
\begin{align}
\frac{1}{\card{\calF}}\sum_{f\in \calF}\frac{1}{\card{\calZ}} \sum_{z\in \calZ} {\norm{\rho_A - \frac{1}{h}\sum_{x\in f^{-1}(z)}  \rho^{x}_A}}_1\leq \delta.
\end{align}
In particular, there exists a function $g:\intseq{1}{h}\to \calX$ such that
\begin{align}
{\norm{\rho_A - \frac{1}{h}\sum_{r=1}^h  \rho^{g(r)}_A}}_1\leq \delta.
\end{align}
\end{lemma}
\begin{proof}
Let us define $\calZ \eqdef \intseq{1}{\frac{\card{\calX}}{h}}$. By Proposition~\ref{lm:hash-exist} and Proposition~\ref{lm:priv-amp},
\begin{align}
\frac{1}{\card{\calF}}\sum_{f\in \calF}{\norm{(\calE^f_{X\to Z}\otimes\id_A) (\rho_{XA})-\mixed{Z}\otimes \rho_A  }}_1 \\
\leq \inf_{\epsilon\geq 0}\left[ 2\epsilon + 2^{-\frac{1}{2} \pr{\mathbb{H}^{\epsilon}_{\min}(X|A)_{\rho} - \log \card{\calZ} } }\right]. \label{eq:priv-amp-app}
\end{align}

By definition of $\calE^f_{X\to Z}$ and $\rho_{XA}$, we have
\begin{align}
(\calE^f_{X\to Z}\otimes\id_A) (\rho_{XA}) 
&= \frac{1}{\card{\calX}} \sum_{x\in \calX} \kb{f(x)}\otimes \rho_A^x\\
&= \frac{1}{\card{\calZ}} \sum_{z\in \calZ} \kb{z}\otimes \pr{\frac{1}{h}\sum_{x\in f^{-1}(z)} \rho_A^x}.
\end{align}
Therefore, we have
\begin{multline}
{\norm{(\calE^f_{X\to Z}\otimes\id_A) (\rho_{XA})-\mixed{Z}\otimes \rho_A  }}_1\\
\begin{split}
&= {\norm{ \frac{1}{\card{\calZ}}\sum_{z\in \calZ} \kb{z} \otimes \pr{\frac{1}{h}\sum_{x\in f^{-1}(z)} \rho_A^{x} -\rho_A} }}_1\\
&= \frac{1}{\card{\calZ}}\sum_{z\in \calZ} {\norm{{\frac{1}{h}\sum_{x\in f^{-1}(z)} \rho_A^{x} -\rho_A} }}_1 
.\end{split}
\label{eq:average-dens}
\end{multline}
Combining \eqref{eq:priv-amp-app} and \eqref{eq:average-dens}, we have for at least one $z\in \calZ$ and at least one $f\in \calF$,
\begin{multline}
{\norm{{\frac{1}{h}\sum_{x\in f^{-1}(z)} \rho_A^{x} -\rho_A} }}_1 \\
\begin{split}
&\leq \inf_{\epsilon\geq 0}\left[ 2\epsilon + 2^{-\frac{1}{2} \pr{\mathbb{H}^{\epsilon}_{\min}(X|A)_{\rho} - \log \card{\calZ} } }\right]\\
& \stackrel{(a)}{=} \inf_{\epsilon\geq 0}\left[ 2\epsilon + 2^{-\frac{1}{2} \pr{\mathbb{H}^{\epsilon}_{\min}(X|A)_{\rho} + \log h - \log \card{\calX}  } }\right]\stackrel{(a)}{\leq} \delta,
\end{split}
\end{multline}
where $(a)$ follows from \eqref{eq:logk-bound}. Taking a bijection $g:\intseq{1}{h} \to f^{-1}(z)$ completes the proof.
\end{proof}



\section{Reducing public communication when $\nnv$ is a power of prime}
\label{sec:anal-quan-res}
 In the next lemma, we show that under our symmetry conditions on $\calF$ and $\rho_{XA}$, the choice of $z$ does not matter.

\begin{lemma}
\label{lm:res-sym}
Suppose that for all $f\in \calF$, $z,z' \in \calZ$, there exist a bijection $\phi:\calX \to \calX$ and unitary $U$ acting on $\calH_A$ (depending on $z$, $z'$, and $f$) such that
\begin{align}
\phi(f^{-1}(z)) &= f^{-1}(z') \label{eq:sym-1}\\
\rho_A^{\phi(x)} &= U\rho_A^x U^\dagger.\label{eq:sym-2}
\end{align}
We then have 
\begin{align}
\label{eq:z-zp-eq}
{\norm{{\frac{1}{h}\sum_{x\in f^{-1}(z)} \rho_A^{x} -\rho_A} }}_1 = {\norm{{\frac{1}{h}\sum_{x\in f^{-1}(z')} \rho_A^{x} -\rho_A} }}_1.
\end{align}
\end{lemma}
\begin{proof}
Note that
\begin{align}
{\norm{{\frac{1}{h}\sum_{x\in f^{-1}(z)} \rho_A^{x} -\rho_A} }}_1 
&= {\norm{{U\pr{\frac{1}{h}\sum_{x\in f^{-1}(z)} \rho_A^{x} -\rho_A}}U^\dagger }}_1\\
&=  {\norm{{{\frac{1}{h}\sum_{x\in f^{-1}(z)} U\rho_A^{x}U^\dagger -U\rho_AU^\dagger}} }}_1\\
&= {\norm{{{\frac{1}{h}\sum_{x\in f^{-1}(z)} \rho_A^{\phi(x)} -U\rho_AU^\dagger}} }}_1\\
&=  {\norm{{{\frac{1}{h}\sum_{x\in f^{-1}(z')} \rho_A^{x} -U\rho_AU^\dagger}} }}_1.
\end{align}
Moreover, we have
\begin{align}
U\rho_AU^\dagger = U\pr{\frac{1}{\card{\calX} }\sum_{x\in \calX} \rho_A^x}U^\dagger
&=\frac{1}{\card{\calX}} \sum_{x\in \calX} U \rho_A^x U^\dagger\\
&=\frac{1}{\card{\calX}} \sum_{x\in \calX} \rho_A^{\phi(x)}\\
&= \frac{1}{\card{\calX}}\sum_{x\in \calX} \rho_A^{x}= \rho_A.
\end{align}
Therefore, we obtain \eqref{eq:z-zp-eq}.
\end{proof}

When $\nnv$ is a power of a prime, we provide an example of two-universal hash functions satisfying the conditions of Lemma~\ref{lm:res-sym}. We assume  in this paragraph only that $\calV = \intseq{0}{\nnv-1}$ to be consistent with the standard notation for finite fields. Note first that $\calV^\ell$ is a field with component-wise addition modulo $\nnv$ and a multiplication operation denoted by $\odot$. We use the short-hand $0^{m}$ for the all-zero sequence of length $m$ and $\cdot | \cdot$ for the concatenation of two sequences. For $k\in \intseq{1}{\ell}$ and $u^\ell \in \calV^\ell$, let $f_{u^\ell}(v^\ell) $ be the first $k$ elements of $u^\ell \odot v^\ell$. By~\cite{bellare2012polynomial}, $\calF=\set{f_{u^\ell}: u^\ell \in \calV^\ell \setminus \set{0^\ell}}$ is a regular two-universal class of hash functions. Moreover, for any $u^\ell \in \calV^\ell \setminus \set{0}$, $z^k, {z'}^k \in \calV^k$, we define $\phi(v^\ell) = (({z'}^k - {z}^k)| 0^{\ell - k})\odot (u^\ell)^{-1} + v^{\ell} $. We show that $\phi$ satisfies \eqref{eq:sym-1} and \eqref{eq:sym-2}. Note that
\begin{align}
\phi\pr{f_{u^\ell}^{-1}(z^k)} 
&= \phi\pr{\set{v^\ell: \exists r^{\ell - k}:z^k|r^{\ell-k} = u^\ell \odot v^{\ell}}}\\
&= \left\{v^\ell +(({z'}^k - {z}^k)| 0^{\ell - k})\odot (u^\ell)^{-1} : \right. \nonumber\\
& \phantom{=====}\left. \exists r^{\ell - k}:z^k|r^{\ell-k} = u^\ell \odot v^{\ell} \right\}\\
&= \left\{ v^\ell  : \exists r^{\ell - k}:z^k|r^{\ell-k} = \right.\nonumber\\
&\phantom{=} \left. u^\ell \odot (v^{\ell} -(({z'}^k - {z}^k)| 0^{\ell - k})\odot (u^\ell)^{-1}) \right\}\\
&= \set{v^\ell  : \exists r^{\ell - k}:{z'}^k|r^{\ell-k} = u^\ell \odot v^\ell }\\
&= f^{-1}_{u^\ell}({z'}^k)
\end{align}
Furthermore, let $U_{\mathrm{CS}}$ be the unitary operation on $\calH_Q^{\proddist m}$ corresponding to cyclic shift of length 1, i.e., $\ket{\phi_1}\otimes  \cdots \otimes \ket{\phi_m} \mapsto \ket{\phi_m} \otimes \ket{\phi_1} \otimes \cdots \ket{\phi_{m-1}}$. By definition of $d(x, v)$ and $\rho_{Q^n}^{v^\ell}$, we have
\begin{align}
\rho_{Q^n}^{v^\ell + {v'}^\ell} = \pr{U_{\mathrm{CS}}^{v'_1} \otimes \cdots \otimes U_{\mathrm{CS}}^{v'_\ell}} \rho_{Q^n}^{v^\ell}  \pr{U_{\mathrm{CS}}^{v'_1} \otimes \cdots \otimes U_{\mathrm{CS}}^{v'_\ell}}^{\dagger},
\end{align}
where $v^\ell + {v'}^\ell$ is modulo $\nnv$. We therefore conclude that \eqref{eq:sym-2} holds.

\section{Proof of Theorem~\ref{th:data-proc-ref}}
\label{sec:data-proc-ref}
To prove Theorem~\ref{th:data-proc-ref}, we need   the  following  tools.

\begin{theorem}\textnormal{(\cite[Theorem 12.1.1]{wilde2013quantum})}
\label{th:rec}
Let $A$ and $B$ be two quantum systems. Let $\rho_A^0$ and $\rho_A^1$ be in $\calD(\calH_A)$ and $\calN:\calD(\calH_A)\to \calD(\calH_B)$ be a quantum channel. There exists a quantum channel $\calR:\calD(\calH_B)\to \calD(\calH_A)$ (depending only on $\calN$ and $\rho_A^0$) such that
\begin{multline}
\D{\rho_A^1}{\rho_A^0} - \D{\calN(\rho_A^1)}{\calN(\rho_A^0)}\\ \geq
 -\log F(\rho_A^1, (\calR \circ \calN)(\rho_A^1))
\end{multline}
and
\begin{align}
(\calR\circ \calN)(\rho_A^0) &= \rho_A^0.
\end{align}
\end{theorem}
\begin{lemma}
Let $A$ and $B$ be two quantum systems such that $A$ is a composition of two sub-systems $A'$ and $A''$.  Let $\rho_A^0$ and $\rho_A^1$ be in $\calD(\calH_A)$ such that for two pure states $\ket{\phi^0}_{A'}$ and $\ket{\phi^1}_{A'}$ in $\calH_{A'}$ and a mixed state $\nu_{A''}$ in $\calD(\calH_{A''})$, we have $C(\phi^x_{A'} \otimes \nu_{A''}, \rho^x_A)  \leq \delta_x$. Let $\calN: \calD(\calH_A)\to \calD(\calH_A)$ be a quantum channels such that $ F(\rho_A^x, \calN(\rho_A^x)) \geq 1-\epsilon_x$. We then have
\begin{multline}
F(\calE(\rho_A^1), \calE(\rho_A^0))\\ \geq   \aleph(\lambda, F(\phi^1_{A'}, \phi^0_{A'})) - 2\sqrt{1-F(\phi^1_{A'}, \phi^0_{A'})}\delta -\delta^2,
\end{multline}
where $\delta \eqdef \sum_x \delta_x$, $\lambda = \sum_x \sqrt{\epsilon_x + 4\sqrt{\epsilon_x}\delta_x + 4 \delta_x^2}$, 
 $\calE$ is a complementary channel to $\calN$.
\label{lm:fid-to-l1}
\end{lemma}

\begin{proof}
See Appendix~\ref{sec:fid-to-l1}.
\end{proof}

We are now ready to provide the proof of Theorem~\ref{th:data-proc-ref}.
\begin{proof}
By Theorem~\ref{th:rec}, there exists a channel $\calR:\calD(E) \to \calD(A)$ such that
\begin{align}
\label{eq:rec-e-chan}\D{\rho_A^1}{\rho_A^0} - \D{\calE(\rho_A^1)}{\calE(\rho_A^0)} &\geq -\log F(\rho_A^1, (\calR \circ \calE)(\rho_A^1))\\
(\calR\circ \calE)(\rho_A^0) &= \rho_A^0.
\end{align}
Let $\calU_{A\to BE}$ be an isometric extension of $\calN$ compatible with $\calE$. Let $\calW_{E\to AF}$ be an isometric extension of $\calR$. The isometry $(\one_B \otimes \calW_{E\to AF})\calU_{A\to BE}$ is an isometric extension of $\calR \circ \calE$. Hence, the mapping 
\begin{multline}
\rho \mapsto \\
\ptr{A}{(\one_B \otimes \calW_{E\to AF})\calU_{A\to BE} \rho ( (\one_B \otimes \calW_{E\to AF})\calU_{A\to BE})^\dagger}
\end{multline}
is a complementary channel of $\calR \circ \calE$ and
\begin{multline}
\ptr{AF}{(\one_B \otimes \calW_{E\to AF})\calU_{A\to BE} \rho ( (\one_B \otimes \calW_{E\to AF})\calU_{A\to BE})^\dagger} \\= \ptr{E}{\calU_{A\to BE} \rho \calU_{A\to BE}} = \calN(\rho)
\end{multline}
Therefore, $\calN$ is a degraded version of the complementary channel of $\calR \circ \calE$. Hence, by Lemma~\ref{lm:fid-to-l1}, we have
\begin{multline}
F(\calN(\rho_A^1), \calN(\rho_A^0)) \geq \aleph(\lambda', F(\phi^1_{A'}, \phi^0_{A'}))\\ - 2\sqrt{1-F(\phi^1_{A'}, \phi^0_{A'})}\delta -\delta^2
\end{multline}
where 
\begin{align}
\lambda'
&\eqdef \sum_{x} \left(1 - F(\rho_A^x, \calR(\calE(\rho_A^x))) \nonumber\right.\\
&\left.\phantom{====}+ 4\sqrt{1 - F(\rho_A^x, \calR(\calE(\rho_A^x))) }\delta_x + 4\delta_x^2\right)^{\frac{1}{2}}\\
&=  2\delta_0  + \left(1 - F(\rho_A^1, \calR(\calE(\rho_A^1))\right.\nonumber\\
&\phantom{=====}\left.+ 4\sqrt{1 - F(\rho_A^1, \calR(\calE(\rho_A^1))}\delta_1 + 4\delta_1^2\right)^\frac{1}{2}.
\end{align}
By our assumption in \eqref{eq:eta-assump}, we have $\aleph(\lambda, F(\phi_A^1, \phi_A^0)) \geq \aleph(\lambda', F(\phi_A^1, \phi_A^0))$. Since $\aleph(x, y)$ is decreasing in $x$ for positive $x$, we have
\begin{align}
\lambda' \geq \lambda,
\end{align}
which yields that $1 - \eta \geq 1 - F(\rho_A^1, \calR(\calE(\rho_A^1))$. Substituting this inequality in \eqref{eq:rec-e-chan} completes the proof of our claim.

\end{proof}
\subsection{Proof of Lemma~\ref{lm:fid-to-l1}}
\label{sec:fid-to-l1}
We first prove a ``triangle'' inequality for fidelity measure, which follows from the triangle inequality for $C(\cdot, \cdot)$.
\begin{lemma}
\label{lm:triangle-fid}
Let $\rho, \sigma, \rho', \sigma' \in \calD(A)$ and let $\epsilon \eqdef C(\rho, \rho') + C(\sigma, \sigma')$. We then have
\begin{align}
F(\rho, \sigma) \geq F(\rho', \sigma') -2\sqrt{1 - F(\rho', \sigma') }\epsilon - \epsilon^2.
\end{align}
\end{lemma}
\begin{proof}
By the triangle inequality for $C(\cdot, \cdot)$, we have
\begin{align}
C(\rho, \sigma) \leq C(\rho', \sigma') + C(\rho, \rho') + C(\sigma, \sigma') = C(\rho', \sigma') + \epsilon
\end{align}
This could be written as
\begin{align}
\sqrt{1-F(\rho, \sigma)} \leq \sqrt{1-F(\rho', \sigma')} + \epsilon.
\end{align}
Therefore,
\begin{align}
1-F(\rho, \sigma) 
&\leq 1- F(\rho', \sigma') + 2\epsilon\sqrt{1-F(\rho', \sigma')} + \epsilon^2,
\end{align}
which yields the desired bound.
\end{proof}
We now prove a result similar to Lemma~\ref{lm:fid-to-l1} when $\rho^0_A$ and $\rho^1_A$ are pure.
\begin{lemma}
\label{lm:pure-dist}
Let $A$ and $B$ be finite dimensional quantum systems such that $A$ is a composition of two sub-systems $A'$ and $A''$. Let $\ket{\phi^0}_{A'}$ and $\ket{\phi^1}_{A'}$ be pure states in $\calH_{A'}$ and $\nu_{A''}$ be a mixed state in $\calD(\calH_{A''})$. Let us define $\rho_A^x \eqdef \phi^x_{A'} \otimes \nu_{A''}$. Let $V:\calH_A\to \calH_A\otimes \calH_B$ be an isometry and define $\psi^x_{AB} \eqdef V\rho_A^x V^\dagger$. Let 
\begin{align}
\epsilon \eqdef \sum_x C(\psi^x_A, \rho_A^x)
\end{align}
 We then have
\begin{align}
F(\psi^1_B, \psi^0_B) \geq  \aleph(\epsilon, F(\phi^1_{A'}, \phi^0_{A'}))
\end{align}
\end{lemma}

\begin{proof}
Let $\ket{\nu}_{RA''}$ be a purification of $\nu_{A''}$ and define $\ket{\psi^x}_{RAB} \eqdef \one_R \otimes V (\ket{\phi^x}_{A'} \otimes \ket{\nu}_{A''R})$ (which is consistent with the definition of $\psi^x_{AB}$). By Uhlmann's theorem, there exist isometries $U^0$ and $U^1$ from $\calH_{R}$ to $\calH_{R} \otimes \calH_{B}$ such that
\begin{align}
C(\psi_A^x, \rho^x_A) = C(\psi_{ABR}^x, \phi^x_{A'}\otimes U^x{\nu}_{A''R}(U^x)^\dagger )
\end{align}

Furthermore, note that
\begin{align}
&F(\phi^1_{A'}, \phi^0_{A'}) \\
&F=(\phi^1_{A'}\otimes \nu_{A''R}, \phi^0_{A'}\otimes\nu_{A''R}) \\
&\stackrel{(a)}{=} F(\psi^1_{ABR}, \psi^0_{ABR}) \\
&\stackrel{(b)}{\leq} F(\phi^1_{A'}\otimes U^1{\nu}_{A''R}(U^1)^\dagger, \phi^0_{A'}\otimes U^0{\nu}_{A''R}(U^0)^\dagger)\nonumber \\
&\phantom{=======}+ 2\sqrt{1- F(\psi_{ABR}^1, \psi^0_{ABR} )} \epsilon + \epsilon^2\\
&= F(\phi^1_{A'}\otimes U^1{\nu}_{A''R}(U^1)^\dagger, \phi^0_{A'}\otimes U^0{\nu}_{A''R}(U^0)^\dagger)\nonumber \\
&\phantom{=======} + 2\sqrt{1- F(\phi_{A}^1, \phi^0_{A} )} \epsilon + \epsilon^2\\
&= F(\phi^1_{A'}, \phi^0_{A'})  F( U^1{\nu}_{A''R}(U^1)^\dagger,  U^0{\nu}_{A''R}(U^0)^\dagger)\nonumber \\
&\phantom{=======} + 2\sqrt{1- F(\phi_{A}^1, \phi^0_{A} )} \epsilon + \epsilon^2,
\end{align}
where $(a)$ follows since $V_{A\to AB}$ is an isometry, and $(b)$ follows from Lemma~\ref{lm:triangle-fid} 
 Therefore, we have
\begin{multline}
 F( U^1{\nu}_{A''R}(U^1)^\dagger,  U^0{\nu}_{A''R}(U^0)^\dagger) \\
 \geq 1 - \frac{2\sqrt{1-F(\phi_A^1, \phi_A^0)}\epsilon + \epsilon^2}{F(\phi^1_A, \phi_A^0)}
\end{multline}
Using Lemma~\ref{lm:triangle-fid} again, we obtain
\begin{align}
F(\psi^1_B, \psi^0_B) 
&\geq F( U^1{\nu}_{A''R}(U^1)^\dagger,  U^0{\nu}_{A''R}(U^0)^\dagger) \nonumber\\
&- 2\sqrt{1-F( U^1{\nu}_{A''R}(U^1)^\dagger,  U^0{\nu}_{A''R}(U^0)^\dagger)}\epsilon - \epsilon^2 \\
&\geq1 - \frac{2\sqrt{1-F(\phi_A^1, \phi_A^0)}\epsilon + \epsilon^2}{F(\phi^1_A, \phi_A^0)}\nonumber\\
&\phantom{===}- 2\sqrt{\frac{2\sqrt{1-F(\phi_A^1, \phi_A^0)}\epsilon + \epsilon^2}{F(\phi^1_A, \phi_A^0)}}\epsilon - \epsilon^2\\
&= \aleph(\epsilon, F(\phi^1_A, \phi^0_A)).
\end{align}
\end{proof}
We now prove Lemma~\ref{lm:fid-to-l1}. Note that for 
\begin{align}
\lambda\eqdef C(\phi^0, \calN(\phi^0)) + C(\phi^1, \calN(\phi^1)),
\end{align}
we have
\begin{align}
&F(\calE(\rho_A^1), \calE(\rho_A^0)) \\
&\stackrel{(a)}{\geq} F(\calE(\phi_A^1), \calE(\phi_A^0)) - 2\sqrt{1 - F(\calE(\phi_A^1), \calE(\phi_A^0))}\delta -\delta^2\\
&\geq F(\calE(\phi_A^1), \calE(\phi_A^0)) - 2\sqrt{1 - F(\phi_A^1, \phi_A^0)}\delta -\delta^2\\
&\stackrel{(b)}{\geq} \aleph(\lambda, F(\phi_A^1, \phi^0_A)),  - 2\sqrt{1 - F(\phi_A^1, \phi_A^0)}\delta -\delta^2,
\end{align}
where $(a)$ follows from Lemma~\ref{lm:triangle-fid}, and $(b)$ follows from Lemma~\ref{lm:pure-dist}
Additionally, we have
\begin{align}
&F(\phi^x, \calN(\phi^x)) \\
&\geq F(\rho^x, \calN(\rho^x)) - 4 \sqrt{1-F(\rho^x, \calN(\rho^x))}\delta_x - 4\delta_x^2\\
&\geq 1 -\epsilon_x -   4 \sqrt{\epsilon_x}\delta_x - 4\delta_x^2,
\end{align}
for $x=0, 1$. This implies that $\lambda \leq \sum_x \sqrt{\epsilon_x + 4\sqrt{\epsilon_x}\delta_x + 4\delta_x^2}$.
\bibliographystyle{apsrev4-1}
\bibliography{covert}

\begin{thebibliography}{29}%
\makeatletter
\providecommand \@ifxundefined [1]{%
 \@ifx{#1\undefined}
}%
\providecommand \@ifnum [1]{%
 \ifnum #1\expandafter \@firstoftwo
 \else \expandafter \@secondoftwo
 \fi
}%
\providecommand \@ifx [1]{%
 \ifx #1\expandafter \@firstoftwo
 \else \expandafter \@secondoftwo
 \fi
}%
\providecommand \natexlab [1]{#1}%
\providecommand \enquote  [1]{``#1''}%
\providecommand \bibnamefont  [1]{#1}%
\providecommand \bibfnamefont [1]{#1}%
\providecommand \citenamefont [1]{#1}%
\providecommand \href@noop [0]{\@secondoftwo}%
\providecommand \href [0]{\begingroup \@sanitize@url \@href}%
\providecommand \@href[1]{\@@startlink{#1}\@@href}%
\providecommand \@@href[1]{\endgroup#1\@@endlink}%
\providecommand \@sanitize@url [0]{\catcode `\\12\catcode `\$12\catcode
  `\&12\catcode `\#12\catcode `\^12\catcode `\_12\catcode `\%12\relax}%
\providecommand \@@startlink[1]{}%
\providecommand \@@endlink[0]{}%
\providecommand \url  [0]{\begingroup\@sanitize@url \@url }%
\providecommand \@url [1]{\endgroup\@href {#1}{\urlprefix }}%
\providecommand \urlprefix  [0]{URL }%
\providecommand \Eprint [0]{\href }%
\providecommand \doibase [0]{http://dx.doi.org/}%
\providecommand \selectlanguage [0]{\@gobble}%
\providecommand \bibinfo  [0]{\@secondoftwo}%
\providecommand \bibfield  [0]{\@secondoftwo}%
\providecommand \translation [1]{[#1]}%
\providecommand \BibitemOpen [0]{}%
\providecommand \bibitemStop [0]{}%
\providecommand \bibitemNoStop [0]{.\EOS\space}%
\providecommand \EOS [0]{\spacefactor3000\relax}%
\providecommand \BibitemShut  [1]{\csname bibitem#1\endcsname}%
\let\auto@bib@innerbib\@empty
\bibitem [{\citenamefont {Diamanti}\ \emph {et~al.}(2016)\citenamefont
  {Diamanti}, \citenamefont {Lo}, \citenamefont {Qi},\ and\ \citenamefont
  {Yuan}}]{Diamanti2016}%
  \BibitemOpen
  \bibfield  {author} {\bibinfo {author} {\bibfnamefont {E.}~\bibnamefont
  {Diamanti}}, \bibinfo {author} {\bibfnamefont {H.-K.}\ \bibnamefont {Lo}},
  \bibinfo {author} {\bibfnamefont {B.}~\bibnamefont {Qi}}, \ and\ \bibinfo
  {author} {\bibfnamefont {Z.}~\bibnamefont {Yuan}},\ }\href {\doibase
  10.1038/npjqi.2016.25} {\bibfield  {journal} {\bibinfo  {journal} {npj
  Quantum Information}\ }\textbf {\bibinfo {volume} {2}} (\bibinfo {year}
  {2016}),\ 10.1038/npjqi.2016.25}\BibitemShut {NoStop}%
\bibitem [{\citenamefont {Bennett~Ch}\ and\ \citenamefont
  {Brassard}(1984)}]{bennett1984quantum}%
  \BibitemOpen
  \bibfield  {author} {\bibinfo {author} {\bibfnamefont {H.}~\bibnamefont
  {Bennett~Ch}}\ and\ \bibinfo {author} {\bibfnamefont {G.}~\bibnamefont
  {Brassard}},\ }in\ \href@noop {} {\emph {\bibinfo {booktitle} {Conf. on
  Computers, Systems and Signal Processing (Bangalore, India, Dec. 1984)}}}\
  (\bibinfo {year} {1984})\ pp.\ \bibinfo {pages} {175--9}\BibitemShut
  {NoStop}%
\bibitem [{\citenamefont {Ekert}(1991)}]{Ekert1991}%
  \BibitemOpen
  \bibfield  {author} {\bibinfo {author} {\bibfnamefont {A.~K.}\ \bibnamefont
  {Ekert}},\ }\href {\doibase 10.1103/PhysRevLett.67.661} {\bibfield  {journal}
  {\bibinfo  {journal} {Phys. Rev. Lett.}\ }\textbf {\bibinfo {volume} {67}},\
  \bibinfo {pages} {661} (\bibinfo {year} {1991})}\BibitemShut {NoStop}%
\bibitem [{\citenamefont {Renner}(2008)}]{renner2008security}%
  \BibitemOpen
  \bibfield  {author} {\bibinfo {author} {\bibfnamefont {R.}~\bibnamefont
  {Renner}},\ }\href@noop {} {\bibfield  {journal} {\bibinfo  {journal}
  {International Journal of Quantum Information}\ }\textbf {\bibinfo {volume}
  {6}},\ \bibinfo {pages} {1} (\bibinfo {year} {2008})}\BibitemShut {NoStop}%
\bibitem [{\citenamefont {Ac\'{\i}n}\ \emph {et~al.}(2007)\citenamefont
  {Ac\'{\i}n}, \citenamefont {Brunner}, \citenamefont {Gisin}, \citenamefont
  {Massar}, \citenamefont {Pironio},\ and\ \citenamefont {Scarani}}]{Acin2007}%
  \BibitemOpen
  \bibfield  {author} {\bibinfo {author} {\bibfnamefont {A.}~\bibnamefont
  {Ac\'{\i}n}}, \bibinfo {author} {\bibfnamefont {N.}~\bibnamefont {Brunner}},
  \bibinfo {author} {\bibfnamefont {N.}~\bibnamefont {Gisin}}, \bibinfo
  {author} {\bibfnamefont {S.}~\bibnamefont {Massar}}, \bibinfo {author}
  {\bibfnamefont {S.}~\bibnamefont {Pironio}}, \ and\ \bibinfo {author}
  {\bibfnamefont {V.}~\bibnamefont {Scarani}},\ }\href {\doibase
  10.1103/PhysRevLett.98.230501} {\bibfield  {journal} {\bibinfo  {journal}
  {Phys. Rev. Lett.}\ }\textbf {\bibinfo {volume} {98}},\ \bibinfo {pages}
  {230501} (\bibinfo {year} {2007})}\BibitemShut {NoStop}%
\bibitem [{\citenamefont {Bash}\ \emph {et~al.}(2013)\citenamefont {Bash},
  \citenamefont {Goeckel},\ and\ \citenamefont {Towsley}}]{Bash2013}%
  \BibitemOpen
  \bibfield  {author} {\bibinfo {author} {\bibfnamefont {B.}~\bibnamefont
  {Bash}}, \bibinfo {author} {\bibfnamefont {D.}~\bibnamefont {Goeckel}}, \
  and\ \bibinfo {author} {\bibfnamefont {D.}~\bibnamefont {Towsley}},\ }\href
  {\doibase 10.1109/JSAC.2013.130923} {\bibfield  {journal} {\bibinfo
  {journal} {{IEEE} {J}ournal of Selected Areas in Communications}\ }\textbf
  {\bibinfo {volume} {31}},\ \bibinfo {pages} {1921} (\bibinfo {year}
  {2013})}\BibitemShut {NoStop}%
\bibitem [{\citenamefont {Wang}\ \emph {et~al.}(2016)\citenamefont {Wang},
  \citenamefont {Wornell},\ and\ \citenamefont {Zheng}}]{Wang2016b}%
  \BibitemOpen
  \bibfield  {author} {\bibinfo {author} {\bibfnamefont {L.}~\bibnamefont
  {Wang}}, \bibinfo {author} {\bibfnamefont {G.~W.}\ \bibnamefont {Wornell}}, \
  and\ \bibinfo {author} {\bibfnamefont {L.}~\bibnamefont {Zheng}},\ }\href
  {\doibase 10.1109/TIT.2016.2548471} {\bibfield  {journal} {\bibinfo
  {journal} {IEEE Trans. Info. Theory}\ }\textbf {\bibinfo {volume} {62}},\
  \bibinfo {pages} {3493} (\bibinfo {year} {2016})}\BibitemShut {NoStop}%
\bibitem [{\citenamefont {Bloch}(2016)}]{Bloch2016a}%
  \BibitemOpen
  \bibfield  {author} {\bibinfo {author} {\bibfnamefont {M.~R.}\ \bibnamefont
  {Bloch}},\ }\href {\doibase 10.1109/TIT.2016.2530089} {\bibfield  {journal}
  {\bibinfo  {journal} {IEEE Trans. Info. Theory}\ }\textbf {\bibinfo {volume}
  {62}},\ \bibinfo {pages} {2334} (\bibinfo {year} {2016})}\BibitemShut
  {NoStop}%
\bibitem [{\citenamefont {Sheikholeslami}\ \emph {et~al.}(2016)\citenamefont
  {Sheikholeslami}, \citenamefont {Bash}, \citenamefont {Towsley},
  \citenamefont {Goeckel},\ and\ \citenamefont {Guha}}]{Sheikholeslami2016}%
  \BibitemOpen
  \bibfield  {author} {\bibinfo {author} {\bibfnamefont {A.}~\bibnamefont
  {Sheikholeslami}}, \bibinfo {author} {\bibfnamefont {B.~A.}\ \bibnamefont
  {Bash}}, \bibinfo {author} {\bibfnamefont {D.}~\bibnamefont {Towsley}},
  \bibinfo {author} {\bibfnamefont {D.}~\bibnamefont {Goeckel}}, \ and\
  \bibinfo {author} {\bibfnamefont {S.}~\bibnamefont {Guha}},\ }in\ \href
  {\doibase 10.1109/ISIT.2016.7541662} {\emph {\bibinfo {booktitle} {Proc. of
  IEEE International Symposium on Information Theory}}}\ (\bibinfo {address}
  {Barcelona, Spain},\ \bibinfo {year} {2016})\ pp.\ \bibinfo {pages}
  {2064--2068}\BibitemShut {NoStop}%
\bibitem [{\citenamefont {Arrazola}\ and\ \citenamefont
  {Scarani}(2016)}]{Arrazola2016}%
  \BibitemOpen
  \bibfield  {author} {\bibinfo {author} {\bibfnamefont {J.~M.}\ \bibnamefont
  {Arrazola}}\ and\ \bibinfo {author} {\bibfnamefont {V.}~\bibnamefont
  {Scarani}},\ }\href {\doibase 10.1103/PhysRevLett.117.250503} {\bibfield
  {journal} {\bibinfo  {journal} {Phys. Rev. Lett.}\ }\textbf {\bibinfo
  {volume} {117}},\ \bibinfo {pages} {250503} (\bibinfo {year} {2016})},\
  \Eprint {http://arxiv.org/abs/1604.05438v3} {1604.05438v3} \BibitemShut
  {NoStop}%
\bibitem [{\citenamefont {Bash}\ \emph {et~al.}(2015)\citenamefont {Bash},
  \citenamefont {Gheorghe}, \citenamefont {Patel}, \citenamefont {Habif},
  \citenamefont {Goeckel}, \citenamefont {Towsley},\ and\ \citenamefont
  {Guha}}]{Bash2015a}%
  \BibitemOpen
  \bibfield  {author} {\bibinfo {author} {\bibfnamefont {B.~A.}\ \bibnamefont
  {Bash}}, \bibinfo {author} {\bibfnamefont {A.~H.}\ \bibnamefont {Gheorghe}},
  \bibinfo {author} {\bibfnamefont {M.}~\bibnamefont {Patel}}, \bibinfo
  {author} {\bibfnamefont {J.~L.}\ \bibnamefont {Habif}}, \bibinfo {author}
  {\bibfnamefont {D.}~\bibnamefont {Goeckel}}, \bibinfo {author} {\bibfnamefont
  {D.}~\bibnamefont {Towsley}}, \ and\ \bibinfo {author} {\bibfnamefont
  {S.}~\bibnamefont {Guha}},\ }\href@noop {} {\bibfield  {journal} {\bibinfo
  {journal} {{N}ature {C}ommunications}\ }\textbf {\bibinfo {volume} {6}},\
  (\bibinfo {year} {2015})}\BibitemShut {NoStop}%
\bibitem [{\citenamefont {Tahmasbi}\ and\ \citenamefont
  {Bloch}(2019)}]{Tahmasbi2018b}%
  \BibitemOpen
  \bibfield  {author} {\bibinfo {author} {\bibfnamefont {M.}~\bibnamefont
  {Tahmasbi}}\ and\ \bibinfo {author} {\bibfnamefont {M.~R.}\ \bibnamefont
  {Bloch}},\ }\href {\doibase 10.1103/PhysRevA.99.052329} {\bibfield  {journal}
  {\bibinfo  {journal} {Phys. Rev. A}\ }\textbf {\bibinfo {volume} {99}},\
  \bibinfo {pages} {052329} (\bibinfo {year} {2019})}\BibitemShut {NoStop}%
\bibitem [{\citenamefont {{Watanabe}}\ and\ \citenamefont
  {{Hayashi}}(2013)}]{Watanabe2013}%
  \BibitemOpen
  \bibfield  {author} {\bibinfo {author} {\bibfnamefont {S.}~\bibnamefont
  {{Watanabe}}}\ and\ \bibinfo {author} {\bibfnamefont {M.}~\bibnamefont
  {{Hayashi}}},\ }in\ \href {\doibase 10.1109/ISIT.2013.6620720} {\emph
  {\bibinfo {booktitle} {2013 IEEE International Symposium on Information
  Theory}}}\ (\bibinfo {address} {Istanbul, Turkey},\ \bibinfo {year} {2013})\
  pp.\ \bibinfo {pages} {2715--2719}\BibitemShut {NoStop}%
\bibitem [{\citenamefont {Kadampot}\ \emph {et~al.}(2018)\citenamefont
  {Kadampot}, \citenamefont {Tahmasbi},\ and\ \citenamefont
  {Bloch}}]{Kadampot2018}%
  \BibitemOpen
  \bibfield  {author} {\bibinfo {author} {\bibfnamefont {I.~A.}\ \bibnamefont
  {Kadampot}}, \bibinfo {author} {\bibfnamefont {M.}~\bibnamefont {Tahmasbi}},
  \ and\ \bibinfo {author} {\bibfnamefont {M.~R.}\ \bibnamefont {Bloch}},\ }in\
  \href {\doibase 10.1109/ISIT.2018.8437587} {\emph {\bibinfo {booktitle}
  {Proc. of IEEE International Symposium on Information Theory}}}\ (\bibinfo
  {address} {Vail, CO},\ \bibinfo {year} {2018})\ pp.\ \bibinfo {pages}
  {1864--1868}\BibitemShut {NoStop}%
\bibitem [{Note1()}]{Note1}%
  \BibitemOpen
  \bibinfo {note} {One can associate a mixed state to no communication, but in
  bosonic systems, the natural choice for the idle state is a pure vacuum
  state.}\BibitemShut {Stop}%
\bibitem [{\citenamefont {{Bloch}}\ and\ \citenamefont
  {{Guha}}(2017)}]{Bloch2017b}%
  \BibitemOpen
  \bibfield  {author} {\bibinfo {author} {\bibfnamefont {M.~R.}\ \bibnamefont
  {{Bloch}}}\ and\ \bibinfo {author} {\bibfnamefont {S.}~\bibnamefont
  {{Guha}}},\ }in\ \href {\doibase 10.1109/ISIT.2017.8007045} {\emph {\bibinfo
  {booktitle} {2017 IEEE International Symposium on Information Theory}}}\
  (\bibinfo {address} {Aachen, Germany},\ \bibinfo {year} {2017})\ pp.\
  \bibinfo {pages} {2825--2829}\BibitemShut {NoStop}%
\bibitem [{\citenamefont {Hayashi}(2006)}]{hayashi2006quantum}%
  \BibitemOpen
  \bibfield  {author} {\bibinfo {author} {\bibfnamefont {M.}~\bibnamefont
  {Hayashi}},\ }\href@noop {} {\emph {\bibinfo {title} {Quantum information}}}\
  (\bibinfo  {publisher} {Springer},\ \bibinfo {year} {2006})\BibitemShut
  {NoStop}%
\bibitem [{\citenamefont {Martinez-Mateo}\ \emph {et~al.}(2013)\citenamefont
  {Martinez-Mateo}, \citenamefont {Elkouss},\ and\ \citenamefont
  {Martin}}]{Mateo2013}%
  \BibitemOpen
  \bibfield  {author} {\bibinfo {author} {\bibfnamefont {J.}~\bibnamefont
  {Martinez-Mateo}}, \bibinfo {author} {\bibfnamefont {D.}~\bibnamefont
  {Elkouss}}, \ and\ \bibinfo {author} {\bibfnamefont {V.}~\bibnamefont
  {Martin}},\ }\href {https://doi.org/10.1038/srep01576
  http://10.0.4.14/srep01576} {\bibfield  {journal} {\bibinfo  {journal}
  {Scientific Reports}\ }\textbf {\bibinfo {volume} {3}},\ \bibinfo {pages}
  {1576} (\bibinfo {year} {2013})}\BibitemShut {NoStop}%
\bibitem [{\citenamefont {{Chou}}\ and\ \citenamefont
  {{Bloch}}(2016)}]{Chou2016}%
  \BibitemOpen
  \bibfield  {author} {\bibinfo {author} {\bibfnamefont {R.~A.}\ \bibnamefont
  {{Chou}}}\ and\ \bibinfo {author} {\bibfnamefont {M.~R.}\ \bibnamefont
  {{Bloch}}},\ }\href {\doibase 10.1109/TIT.2016.2539145} {\bibfield  {journal}
  {\bibinfo  {journal} {IEEE Trans. Info. Theory}\ }\textbf {\bibinfo {volume}
  {62}},\ \bibinfo {pages} {2410} (\bibinfo {year} {2016})}\BibitemShut
  {NoStop}%
\bibitem [{\citenamefont {Tahmasbi}\ and\ \citenamefont
  {Bloch}(2018)}]{tahmasbi2018framework}%
  \BibitemOpen
  \bibfield  {author} {\bibinfo {author} {\bibfnamefont {M.}~\bibnamefont
  {Tahmasbi}}\ and\ \bibinfo {author} {\bibfnamefont {M.~R.}\ \bibnamefont
  {Bloch}},\ }\href@noop {} {\bibfield  {journal} {\bibinfo  {journal} {arXiv
  preprint arXiv:1811.05626}\ } (\bibinfo {year} {2018})}\BibitemShut {NoStop}%
\bibitem [{Note2()}]{Note2}%
  \BibitemOpen
  \bibinfo {note} {In the classical setting, the authors of~\cite {Bloch2016b}
  showed the upper-bound with a factor of $1/2$ on the right hand side. While
  we conjecture that an extension of such upper-bound to the quantum setting is
  possible, we could only prove the upper-bound without the factor
  $1/2$}\BibitemShut {NoStop}%
\bibitem [{Note3()}]{Note3}%
  \BibitemOpen
  \bibinfo {note} {For any $h'$, we can choose $h$ such that $\protect
  \ensuremath {\left |{{\protect \mathcal {V}}}\right |}h'\geqslant h\geqslant
  h'$ and $\protect \ensuremath {\left |{{\protect \mathcal {V}}}\right |}^\ell
  $ is divisible by $h$; hence, this condition adds at most $\protect \qopname
  \relax o{log}\protect \ensuremath {\left |{{\protect \mathcal {V}}}\right |}
  = O(\protect \qopname \relax o{log}n)$ of penalty on $\protect \qopname
  \relax o{log}h$.}\BibitemShut {Stop}%
\bibitem [{\citenamefont {Furrer}\ \emph {et~al.}(2012)\citenamefont {Furrer},
  \citenamefont {Franz}, \citenamefont {Berta}, \citenamefont {Leverrier},
  \citenamefont {Scholz}, \citenamefont {Tomamichel},\ and\ \citenamefont
  {Werner}}]{Franz2012}%
  \BibitemOpen
  \bibfield  {author} {\bibinfo {author} {\bibfnamefont {F.}~\bibnamefont
  {Furrer}}, \bibinfo {author} {\bibfnamefont {T.}~\bibnamefont {Franz}},
  \bibinfo {author} {\bibfnamefont {M.}~\bibnamefont {Berta}}, \bibinfo
  {author} {\bibfnamefont {A.}~\bibnamefont {Leverrier}}, \bibinfo {author}
  {\bibfnamefont {V.~B.}\ \bibnamefont {Scholz}}, \bibinfo {author}
  {\bibfnamefont {M.}~\bibnamefont {Tomamichel}}, \ and\ \bibinfo {author}
  {\bibfnamefont {R.~F.}\ \bibnamefont {Werner}},\ }\href {\doibase
  10.1103/PhysRevLett.109.100502} {\bibfield  {journal} {\bibinfo  {journal}
  {Phys. Rev. Lett.}\ }\textbf {\bibinfo {volume} {109}},\ \bibinfo {pages}
  {100502} (\bibinfo {year} {2012})}\BibitemShut {NoStop}%
\bibitem [{\citenamefont {Leverrier}\ \emph {et~al.}(2013)\citenamefont
  {Leverrier}, \citenamefont {Garc\'{\i}a-Patr\'on}, \citenamefont {Renner},\
  and\ \citenamefont {Cerf}}]{Leverrier2013}%
  \BibitemOpen
  \bibfield  {author} {\bibinfo {author} {\bibfnamefont {A.}~\bibnamefont
  {Leverrier}}, \bibinfo {author} {\bibfnamefont {R.}~\bibnamefont
  {Garc\'{\i}a-Patr\'on}}, \bibinfo {author} {\bibfnamefont {R.}~\bibnamefont
  {Renner}}, \ and\ \bibinfo {author} {\bibfnamefont {N.~J.}\ \bibnamefont
  {Cerf}},\ }\href {\doibase 10.1103/PhysRevLett.110.030502} {\bibfield
  {journal} {\bibinfo  {journal} {Phys. Rev. Lett.}\ }\textbf {\bibinfo
  {volume} {110}},\ \bibinfo {pages} {030502} (\bibinfo {year}
  {2013})}\BibitemShut {NoStop}%
\bibitem [{\citenamefont {Renner}\ and\ \citenamefont
  {Cirac}(2009)}]{Renner2009}%
  \BibitemOpen
  \bibfield  {author} {\bibinfo {author} {\bibfnamefont {R.}~\bibnamefont
  {Renner}}\ and\ \bibinfo {author} {\bibfnamefont {J.~I.}\ \bibnamefont
  {Cirac}},\ }\href {\doibase 10.1103/PhysRevLett.102.110504} {\bibfield
  {journal} {\bibinfo  {journal} {Phys. Rev. Lett.}\ }\textbf {\bibinfo
  {volume} {102}},\ \bibinfo {pages} {110504} (\bibinfo {year}
  {2009})}\BibitemShut {NoStop}%
\bibitem [{\citenamefont {Bloch}\ and\ \citenamefont
  {Guha}(2017)}]{Bloch2016b}%
  \BibitemOpen
  \bibfield  {author} {\bibinfo {author} {\bibfnamefont {M.}~\bibnamefont
  {Bloch}}\ and\ \bibinfo {author} {\bibfnamefont {S.}~\bibnamefont {Guha}},\
  }\href@noop {} {\enquote {\bibinfo {title} {{Pulse Position Modulation-Based
  Covert Communications}},}\ }\bibinfo {howpublished} {accepted in \emph{IEEE
  International Symposium on Information Theory}} (\bibinfo {year}
  {2017})\BibitemShut {NoStop}%
\bibitem [{\citenamefont {{Ahlswede}}\ and\ \citenamefont
  {{Csiszar}}(1993)}]{Ahlswede1993}%
  \BibitemOpen
  \bibfield  {author} {\bibinfo {author} {\bibfnamefont {R.}~\bibnamefont
  {{Ahlswede}}}\ and\ \bibinfo {author} {\bibfnamefont {I.}~\bibnamefont
  {{Csiszar}}},\ }\href {\doibase 10.1109/18.243431} {\bibfield  {journal}
  {\bibinfo  {journal} {IEEE Transactions on Information Theory}\ }\textbf
  {\bibinfo {volume} {39}},\ \bibinfo {pages} {1121} (\bibinfo {year}
  {1993})}\BibitemShut {NoStop}%
\bibitem [{\citenamefont {Wilde}(2013)}]{wilde2013quantum}%
  \BibitemOpen
  \bibfield  {author} {\bibinfo {author} {\bibfnamefont {M.~M.}\ \bibnamefont
  {Wilde}},\ }\href@noop {} {\emph {\bibinfo {title} {Quantum information
  theory}}}\ (\bibinfo  {publisher} {Cambridge University Press},\ \bibinfo
  {year} {2013})\BibitemShut {NoStop}%
\bibitem [{\citenamefont {Bellare}\ and\ \citenamefont
  {Tessaro}(2012)}]{bellare2012polynomial}%
  \BibitemOpen
  \bibfield  {author} {\bibinfo {author} {\bibfnamefont {M.}~\bibnamefont
  {Bellare}}\ and\ \bibinfo {author} {\bibfnamefont {S.}~\bibnamefont
  {Tessaro}},\ }\href@noop {} {\bibfield  {journal} {\bibinfo  {journal} {arXiv
  preprint arXiv:1201.3160}\ } (\bibinfo {year} {2012})}\BibitemShut {NoStop}%
\end{thebibliography}%

\end{document}